\documentclass[11pt]{article}

\usepackage[sort,numbers]{natbib}
\usepackage[bookmarks=true,bookmarksnumbered=true,colorlinks=true,allcolors=black]{hyperref}

\usepackage{graphicx}
\usepackage{amsmath,amssymb,amsfonts,amsthm}
\usepackage{array}   
\newcolumntype{C}{>{$}c<{$}} 
\usepackage{lscape}
\usepackage{arydshln}

\renewcommand*{\baselinestretch}{1.25}

\DeclareMathAlphabet{\mathcal}{OMS}{cmsy}{m}{n}

\usepackage{geometry}
\usepackage{indentfirst}

\usepackage{enumitem}

\newtheorem{theorem}{Theorem}
\newtheorem{lemma}{Lemma}
\newtheorem{proposition}{Proposition}

\theoremstyle{definition}

\newtheorem*{rmk*}{Remark}
\newtheorem{rmk}{Remark}

\renewcommand*\proofname{\upshape{\bfseries{Proof}}}

\makeatletter
\renewenvironment{proof}[1][\proofname]{\par
  \pushQED{\qed}%
  \normalfont \topsep6\p@\@plus6\p@\relax
  \trivlist
  \item[\hskip\labelsep
        \bfseries
    #1\@addpunct{.}]\ignorespaces
}{%
  \popQED\endtrivlist\@endpefalse
}
\makeatother

\DeclareMathOperator{\variance}{Var}
\DeclareMathOperator{\covariance}{Cov}

\DeclareMathOperator{\trace}{tr}

\allowdisplaybreaks

\makeatletter
    \renewcommand*{\section}{\@startsection{section}{1}{\z@}%
    {6pt}{3pt}{\reset@font\normalsize\bfseries}}
\makeatother
    
\makeatletter
    \renewcommand*{\subsection}{\@startsection{subsection}{2}{\z@}%
    {3pt}{3pt}{\reset@font\normalsize\mdseries\itshape}}
\makeatother

\makeatletter
    \renewcommand*{\subsubsection}{\@startsection{subsubsection}{3}{\z@}%
    {3pt}{3pt}{\reset@font\normalsize\mdseries\itshape}}
\makeatother

\makeatletter
\def\@listi{\leftmargin\leftmargini
  \topsep=.5\baselineskip 
  \partopsep=0pt \parsep=0pt \itemsep=0pt}
\let\@listI\@listi
\@listi
\def\@listii{\leftmargin\leftmarginii
  \labelwidth\leftmarginii \advance\labelwidth-\labelsep
  \topsep=0pt \partopsep=0pt \parsep=0pt \itemsep=0pt}
\def\@listiii{\leftmargin\leftmarginiii
  \labelwidth\leftmarginiii \advance\labelwidth-\labelsep
  \topsep=0pt \partopsep=0pt \parsep=0pt \itemsep=0pt}
\def\@listiv{\leftmargin\leftmarginiv
  \labelwidth\leftmarginiv \advance\labelwidth-\labelsep
  \topsep=0pt \partopsep=0pt \parsep=0pt \itemsep=0pt}
\makeatother

\title{Wavelet-based methods for high-frequency lead-lag analysis}
\author{Takaki Hayashi\thanks{Keio University, Graduate School of Business Administration, 4-1-1 Hiyoshi, Yokohama, Kanagawa 223-8526, Japan}
\thanks{Tokyo Metropolitan University, Department of Business Administration, Graduate School of Social Sciences, Marunouchi Eiraku Bldg. 18F, 1-4-1 Marunouchi, Chiyoda-ku, Tokyo 100-0005 Japan}
\thanks{CREST, Japan Science and Technology Agency}
\and
Yuta Koike\thanks{Mathematics and Informatics Center and Graduate School of Mathematical Sciences, The University of Tokyo, 3-8-1 Komaba, Meguro-ku, Tokyo 153-8914 Japan}
\footnotemark[2]
\thanks{The Institute of Statistical Mathematics, 10-3 Midori-cho, Tachikawa, Tokyo 190-8562, Japan}
\footnotemark[3]}

\begin{document}

\maketitle

\begin{abstract}

We propose a novel framework to investigate lead-lag relationships between two financial assets. Our framework bridges a gap between continuous-time modeling based on Brownian motion and the existing wavelet methods for lead-lag analysis based on discrete-time models and enables us to analyze the multi-scale structure of lead-lag effects. 
We also present a statistical methodology for the scale-by-scale analysis of lead-lag effects in the proposed framework and develop an asymptotic theory applicable to a situation including stochastic volatilities and irregular sampling.   
Finally, we report several numerical experiments to demonstrate how our framework works in practice.     
\vspace{3mm}

\noindent \textit{Keywords}: Brownian motion; High-frequency data; Lead-lag effect; Multiscale modeling; Wavelet.

\end{abstract}

\section{Introduction}

A lead-lag effect is a phenomenon where one asset (called a ``leader'') is correlated with another asset (called a ``lagger'') at later times. 
Investigation of such a phenomenon has a long history in the economics literature; it is measured at various time scales mostly dailies or longer than daily, but relatively less for intra-daily data. It is not surprising that different lead-lag effects can be observed at different time scales in the financial markets, for there are a variety of participants in financial markets: They have different views for the economy and markets, different risk attitudes, with different sources of money and information, which makes them have different investment/trading horizons (cf.~\citet{MDDOPW1997}).

The aim of this paper is to capture such multi-scale structures of lead-lag effects inherent in high-frequency financial data. The wavelet analysis provides a canonical framework to accomplish this purpose. The application of the wavelet analysis {in} finance has recently been expanded in various ways. We refer to the book \citet{GSW2002} for an introduction of such applications. The application of the wavelet analysis to {exploring} lead-lag relationships in the financial markets has also been investigated by many researchers {in most cases} \textit{with middle or low-frequency data}; see \cite{Gallegati2008, Cardinali2009,Ranta2010,FM2012,Dajcman2013} as well as Chapter 7 of \cite{GSW2002} for example. {In the meantime}, there are few articles applying it to the investigation of lead-lag effects in the high-frequency domain. One {of such} exception{s} is the paper by \citet{Hafner2012}, which explores the lead-lag relation between the returns, durations and volumes of high-frequency transaction data of the IBM {stock}. 

To our knowledge, all the existing studies on the wavelet analysis of lead-lag effects have their theoretical basis on discrete-time process modeling, which is mainly pioneered in \citet{WGP1999,WGP2000} and \citet{SW2000a,SW2000b}. 
This is presumably because few results are available for statistical modeling of lead-lag effects in discretely observed continuous-time processes, even without taking their multi-scale structure into account.  
In the meantime, the modern, continuous-time finance theory is based on semimartingale processes, especially driven by Brownian motions (cf.~\citet{Duffie1991}). {Besides}, it has been getting recognized that modeling high-frequency financial data as discrete observations of continuous-time processes is {powerful} for the statistical analysis (cf.~\citet{AJ2014}). From these perspectives it is desirable to develop a class of multivariate models to incorporate lead-lag relationships between coordinates for Brownian motion driven models, which is the primary goal of this study. 
We shall remark that some authors have recently developed statistical modeling of lead-lag effects in the continuous-time framework. \citet{HRY2013} have introduced a simple model to describe lead-lag effects between two continuous It\^o semimartingales and developed a statistical estimation method for the lead-lag effects from (possibly asynchronous) discrete observation data. A similar model is studied in \citet{RR2010} with {a} different methodology. Empirical applications of \citet{HRY2013}'s methodology are found in \cite{AM2014,HA2014,CCI2016,BOW2016}. 
Apart from Brownian motion driven models, the Hawkes processes {may} be a credible candidate to describe lead-lag effects in the continuous-time framework; see \citet{BDHM2013} and \citet{DaFZ2016} for example.   
However, none of them takes account of potential multi-scale structure of lead-lag effects. This paper attempts to fill in this {missing part}. 

Based on the proposed continuous-time model, we also aim at developing a statistical theory for the scale-by-scale analysis of lead-lag effects from discrete observation data.
Our paper is virtually the first attempt to bridge the gap between the two distinct areas of research, namely wavelet analysis and continuous-time stochastic processes, {in the context of lead-lag analysis\footnote{{We note that there are a number of articles dealing with wavelet analysis and stochastic processes in other fields; see e.g.~Chapter 9 of \cite{Vidakovic1999} and references therein (see also \citet{HMSH2012} and references therein for studies in the context of high-frequency financial data).}}}. 
{To complement the theoretical results, we apply the proposed method to both simulated and empirical data and show its effectiveness in practical applications. 
In particular, our empirical application reveals that there is a striking multi-scale structure in the lead-lag relation between the S\&P500 index and the futures, which is indeed in line with the so-called heterogeneous market hypothesis (cf.~\cite{MDDOPW1997}).} 



The paper is organized as follows. Section \ref{section:model} presents two frameworks for scale-by-scale modeling of lead-lag effects between two Brownian motions. Section \ref{theory} develops an asymptotic theory for the {second} framework presented in Section \ref{section:model}. We illustrate {a} numerical performance of the proposed approach on simulated data in Sections \ref{simulation} and on real data in Section \ref{empirical}, respectively. Most of the proofs are given in Section \ref{proofs}.

\section*{Notation}

For every $p\in[1,\infty]$, $L^p(\mathbb{R})$ denotes the $L^p$-space of all \textit{complex-valued} Lebesgue measurable functions on $\mathbb{R}$.  
For any function $f\in L^1(\mathbb{R})$, we denote by $\mathcal{F}f$ and $\mathcal{F}^{-1}f$ its Fourier transform and inverse Fourier transform, respectively. Specifically, they are defined by the following formulae:
\[
(\mathcal{F}f)(\lambda)=\int_{-\infty}^\infty f(t)e^{-\sqrt{-1}\lambda t}dt,\qquad
(\mathcal{F}^{-1}f)(t)=\frac{1}{2\pi}\int_{-\infty}^\infty f(\lambda)e^{\sqrt{-1}\lambda t}d\lambda.
\]
The above definition can also be applied to a function $f\in L^2(\mathbb{R})$ by understanding the convergences of the integrals in $L^2(\mathbb{R})$. 
The Fourier inversion formula then reads as 
$
f=\mathcal{F}^{-1}(\mathcal{F}f)
$ in $L^2(\mathbb{R})$
for all $f\in L^2(\mathbb{R})$. Also, the Parseval identity and the convolution theorem read as
\[
\int_{-\infty}^\infty(\mathcal{F}f)(\lambda)\overline{(\mathcal{F}g)(\lambda)}d\lambda=2\pi\int_{-\infty}^\infty f(t)\overline{g(t)}dt
\]
and
\[
\mathcal{F}(f*g)=(\mathcal{F}f)\cdot(\mathcal{F}g)\qquad\text{in }L^2(\mathbb{R})
\]
for all $f,g\in L^2(\mathbb{R})$. 

{For a positive integer $d$, $\mathsf{{E}}_d$ denotes the identity matrix of size $d$.} 
{For a matrix $A$, we denote by $A^\top$ the transpose of $A$. Also, $\|A\|_F$ denotes the Frobenius norm of $A$, i.e.~$\|A\|_F:=\trace(A^\top A)$. }

{The notation $\to^p$ stands for convergence in probability.}

\section{Modeling scale-by-scale lead-lag effects in the continuous-time framework}\label{section:model}

In this section we propose two frameworks to introduce multiple lead-lag relationships between two Brownian motions $B^1_t$ and $B^2_t$ on a scale-by-scale basis. We also present sensible cross-covariance estimators constructed from discrete observations of the processes $B^1_t$ and $B^2_t$ on a fixed interval, which can be used for identifying lead-lag effects scale-by-scale. Specifically, we assume that $B_t=(B_t^1,B_t^2)$ is observed at the time points $i\tau_J$ for $i=0,1,\dots,n$. Here, the unit time $\tau_J$ corresponds to the finest observable resolution, while $n$ is the number of the unit times contained in the sampling period. So, the interval $[0,n\tau_J]$ corresponds to the sampling period. We take $\tau_J=2^{-J-1}$ to make the interpretation of wavelet analysis easier. 

We denote by $(\Omega,\mathcal{F},P)$ the probability space where $B$ is defined.

\subsection{L\'evy-Ciesielski's construction revisited}\label{section:levy}

We first revisit the classical L\'evy-Ciesielski construction of Brownian motion (see \cite{GBJ2016} for a review on this topic). For each $\nu=1,2$, let $\xi^\nu_0$ be a standard normal variable. Also, let $(\xi^\nu_{jk})_{j,k=0}^\infty$ be i.i.d.~standard normal variables independent of $\xi^\nu_0$. Then we define the process $B^\nu=(B^\nu_t)_{t\in[0,1]}$ by
\begin{equation}\label{levy}
B^\nu_t=\xi_0^\nu t+\sum_{j=0}^\infty\sum_{k=0}^{2^j-1}\xi^\nu_{jk}\int_0^t\psi_{jk}(s)ds,
\qquad\nu=1,2,
\end{equation}
where $\psi_{jk}$'s are the Haar functions defined by $\psi_{jk}(s)=2^{j/2}\psi(2^js-k)$, where
\[
\psi(s)=
\left\{
\begin{array}{ll}
1  &  0\leq s<\frac{1}{2}, \\
-1  & \frac{1}{2}\leq s<1, \\
0  &  \text{otherwise}.  
\end{array}
\right.
\]
It is well-known that the infinite sum in the right side of \eqref{levy} has the limit in $C[0,1]$ a.s.~as the function of $t\in[0,1]$, and the limit process $B^\nu$ is a standard Brownian motion. The convergence is also valid in $L^2(P)$ for any $t\in[0,1]$.   

Decomposition \eqref{levy} naturally suggests that we could think that the process
\begin{align*}
B^\nu(j)_t=\sum_{k=0}^{2^j-1}\xi^\nu_{jk}\int_0^t\psi_{jk}(s)ds,\qquad\nu=1,2;j=0,1,2,\dots,
\end{align*}
as a ``Brownian component at the level $j$''. Here, the level $j$ has the unit resolution $\tau_j=2^{-j-1}$. This suggests that we could assess the lead-lag effect at the resolution $\tau_j$ by measuring that between $B^1(j)$ and $B^2(j)$ in some sense. A standard way to measure the lead-lag effect between two processes is assessing the cross-covariance function of their returns, provided that they are jointly stationary. However, this approach cannot be applied to the processes $B^1(j)$ and $B^2(j)$ directly because they are not of stationary increments. Instead, we propose assessing the cross-covariance function between $(\xi^1_{jk})_{k=0}^\infty$ and $(\xi^2_{jk})_{k=0}^\infty$, provided that {they are jointly (weakly) stationary}. Specifically, the objectives are given by
\begin{equation}\label{dwt-corr}
\rho^0_j(2\cdot l\tau_j)=\covariance\left(\xi^1_{jk},\xi^2_{jk+l}\right),\qquad l=0,\pm1,\pm2,\dots.
\end{equation}
Here, we write the cross-covariance between $\xi^1_{jk}$ and $\xi^2_{jk+l}$ with respect to the lag $2\cdot l\tau_j$ instead of $l$ alone to emphasize that their physical time difference is $2\cdot l\tau_j$.

Since we can reproduce $\xi^\nu_{jk}$'s for $j\leq J$ via the identity
\begin{equation*}
\xi^\nu_{jk}=\int_0^1\psi_{jk}(s)dB^\nu_s=2^{\frac{j}{2}}\left(2B^\nu_{(2k+1)\tau_{j}}-B^\nu_{(2k+2)\tau_{j}}-B^\nu_{2k\tau_{j}}\right),
\end{equation*}
we can naturally use the following estimator for $\rho_j(2\cdot l\tau_j)$:
\begin{align*}
\widehat{\rho}_{j}^{0}(2\cdot l\tau_j)&=\frac{1}{M-l}\sum_{k=0}^{M-1-l}\left(\int_0^1\psi_{jk}(s)dB^1_s\right)\left(\int_0^1\psi_{jk+l}(s)dB^2_s\right),
\end{align*}
where $M=\lfloor n/2^{J-j+1}\rfloor$. 
This class of estimators has an advantage in terms of computation, for the variables $\int_0^1\psi_{jk}(s)dB^\nu_s$ are known as the \textit{discrete wavelet transform} (DWT) of $(B^\nu_{(k+1)\tau_J}-B^\nu_{k\tau_J})_{k=0}^{n-1}$ (ignoring the boundary variables) and fast computation algorithms for them are known (``pyramid algorithm'', cf.~Section 7.3.1 of \citet{Mallat2009}). 
{On the other hand, the main disadvantage of this approach is that we can define the cross-covariance function only at lags of the form $2^j\cdot l\tau_{J}$ for $l\in\mathbb{Z}$ and $j\leq J$, which is an undesirable restriction on lead-lag analyses. In particular, time-lags that are odd multiples of $\tau_J$ shall not be allowed in this approach, which appears an artificial model assumption.} 
In the next subsection we propose another framework to deal with this issue.  

\subsection{L\'evy-Ciesielski's construction based on dyadic wavelet transform}

The major drawback of the first approach is that we cannot define the cross-covariance at {every discrete grid point} generated for the finest observation resolution $\tau_J$. To overcome this issue, we introduce an alternative decomposition to \eqref{levy}. 
For this purpose we reinterpret the L\'evy-Ciesielski construction of Brownian motion as follows. 
Since $\phi:=1_{[0,1)}$ and $\psi_{jk}$ ($j=0,1,\dots,k=0,1,\dots,2^j-1$) constitute an orthonormal basis of $L^2[0,1]$, we have the following expansion for any $f\in L^2[0,1]$:
\begin{equation}\label{expansion}
f=\langle f,\phi\rangle\phi+\sum_{j=0}^\infty\sum_{k=0}^{2^j-1}\langle f,\psi_{jk}\rangle\psi_{jk}\qquad\text{in }L^2[0,1],
\end{equation}
where $\langle\cdot,\cdot\rangle$ denotes the inner product of $L^2[0,1]$. This implies that
\begin{equation}\label{expansion2}
\int_0^1f(s)dB_s^\nu=\langle f,\phi\rangle B_1^\nu+\sum_{j=0}^\infty\sum_{k=0}^{2^j-1}\langle f,\psi_{jk}\rangle\int_0^1\psi_{jk}(s)dB_s^\nu
\end{equation}
for $\nu=1,2$. Substituting $f=1_{(0,t]}$ in the above equation, we recover the L\'evy-Ciesielski construction \eqref{levy} with $\xi_0^\nu=B^\nu_1$ and $\xi_{jk}^\nu=\int_0^1\psi_{jk}(s)dB_s^\nu$. This suggests that we could obtain a decomposition of Brownian motion analogous to the L\'evy-Ciesielski construction once we have a ``canonical'' decomposition for any $f\in L^2[0,1]$ such as \eqref{expansion}. Motivated by this idea, we consider alternative wavelet decomposition for functions which is suitable for {the current} purpose. 

We begin by recalling that decomposition \eqref{expansion} can be {regarded} as a discretization of Calder\'on's reproducing identity for $f\in L^2(\mathbb{R})$ (cf.~Sections 3.1--3.2 of \cite{Vidakovic1999}):
\begin{equation}\label{calderon}
f(t)
=\frac{1}{C_\psi}\int_{0}^\infty \left[\int_{-\infty}^\infty(W_af)(b)\psi\left(\frac{t-b}{a}\right)db\right]\frac{1}{a^2}da\quad\text{in }L^2(\mathbb{R}),
\end{equation}
where $*$ denotes convolution and we set 
\[
C_\psi=\int_0^\infty\frac{|(\mathcal{F}\psi)(\lambda)|^2}{\lambda}d\lambda
\]
as well as we define the function $W_af:\mathbb{R}\to\mathbb{R}$, which is called the \textit{continuous wavelet transform} of $f$, by
\[
W_af(b)=a^{-\frac{1}{2}}\int_{-\infty}^\infty f(t)\psi\left(\frac{t-b}{a}\right)dt,\qquad
b\in\mathbb{R}. 
\]
In fact, decomposition \eqref{expansion} is obtained by discretizing formula \eqref{calderon} in both the scale parameter $a$ and the shift parameter $b$. Now, what is unsuitable for us in \eqref{expansion2} is that we can only consider discretized time shifts of the forms $2\cdot l\tau_j$ ($l\in\mathbb{Z}$). A natural solution of this issue is to only discretize the scale parameter $a$. This leads to the following expansion for $f\in L^2(\mathbb{R})$:   
\begin{equation}\label{LP}
f=(f*\underline{\phi})*\phi+\sum_{j=0}^\infty 2^j(f*\underline{\psi_j})*\psi_j
\qquad\text{in }L^2(\mathbb{R}).
\end{equation}
Here, for any function $g$ on $\mathbb{R}$ we define the functions $\underline{g}$ and $g_j$ ($j\in\mathbb{Z}$) on $\mathbb{R}$ by setting $\underline{g}(t)=g(-t)$ and $g_j(t)=2^{j/2}g(2^jt)$ for $t\in\mathbb{R}$. Note that $f*\underline{\psi_j}=W_{2^{-j}}f$. Decomposition \eqref{LP} is indeed valid for any $f\in L^2(\mathbb{R})$ by Theorem 5.11 from \cite{Mallat2009}, for we can deduce
\[
|(\mathcal{F}\phi)(\lambda)|^2+\sum_{j=-\infty}^\infty2^j|(\mathcal{F}\psi_j)(\lambda)|^2=1
\]
from the proof of Theorem 5.13 from \cite{Mallat2009}. The corresponding L\'evy-Ciesielski type construction is given as follows:
\begin{proposition}\label{prop:dyadic}
Let $B=(B_t)_{t\in\mathbb{R}}$ be a two-sided Brownian motion. Suppose that real-valued functions $\tilde{\phi},\tilde{\psi}\in L^2(\mathbb{R})$ satisfy
\begin{equation}\label{dyadic}
|(\mathcal{F}\tilde{\phi})(\lambda)|^2+\sum_{j=0}^\infty2^j|(\mathcal{F}\tilde{\psi}_j)(\lambda)|^2=1
\end{equation}
for any $\lambda\in\mathbb{R}$. Then we have
\begin{equation}\label{bm-dyadic}
B_t=\int_0^t\tilde{\xi}*\tilde{\phi}(s)ds+\sum_{j=0}^\infty2^j\int_0^t\tilde{\xi}_j*\tilde{\psi}_{j}(s)ds\qquad\text{in }L^2({P})
\end{equation}
for any $t\in\mathbb{R}$, where
\[
\tilde{\xi}(u)=\int_{-\infty}^\infty\tilde{\phi}(s-u)dB_s,\qquad
\tilde{\xi}_j(u)=\int_{-\infty}^\infty\tilde{\psi}_{j}(s-u)dB_s
\]
{
and the integrals $\int_0^t\tilde{\xi}*\tilde{\phi}(s)ds$ and $\int_0^t\tilde{\xi}_j*\tilde{\psi}_{j}(s)ds$ are interpreted as
\begin{align*}
\int_0^t\tilde{\xi}*\tilde{\phi}(s)ds&:=\lim_{A\to\infty}\int_0^t\left\{\int_{-A}^A \tilde{\xi}(u)\tilde{\phi}(s-u)du\right\}ds,\\
\int_0^t\tilde{\xi}_j*\tilde{\psi}_j(s)ds&:=\lim_{A\to\infty}\int_0^t\left\{\int_{-A}^A \tilde{\xi}_j(u)\tilde{\psi}_j(s-u)du\right\}ds,
\end{align*}
where the limits are taken in $L^2(P)$. 
}
\end{proposition}

A proof is given in Section \ref{proof:dyadic}. This decomposition suggests that we might think that the process $\tilde{\xi}_j(\cdot)$ as an alternative ``Brownian component at the level $j$''. However, unlike the original L\'evy-Ciesielski construction, we do not generally have the independence of Brownian components across different levels, i.e.~the processes $\tilde{\xi}_j(\cdot)$ and $\tilde{\xi}_{j'}(\cdot)$ are not independent for $j\neq j'$, especially when we adopt the Haar wavelets as $\tilde{\phi}=\phi$ and $\tilde{\psi}=\psi$. The lack of such independence makes it challenging to model/interpret lead-lag effects at different levels. 
We can avoid this issue by alternatively adopting band-limited wavelets (i.e.~wavelets having compactly supported Fourier transforms). Specifically, we take the Littlewood-Paley wavelets as follows:\footnote{See page 115 of \citet{Daubechies1992}.}
\begin{align*}
\tilde{\phi}(s)=\phi^{LP}(s):=\frac{\sin(\pi s)}{\pi s},\qquad
\tilde{\psi}(s)=\psi^{LP}(s):=2\phi^{LP}\left(2s\right)-\phi^{LP}(s).
\end{align*}
We may regard the Littlewood-Paley wavelets as the ``representative'' band-limited wavelets because any band-limited function can be recovered from its discrete samples with interpolation based on the Littlewood-Paley scaling function $\phi^{LP}$, according to the Shannon-Whittaker sampling theorem (cf.~Theorem 3.2 of \cite{Mallat2009}). 
Now, since we have
\[
(\mathcal{F}\phi^{LP})(\lambda)=1_{[-\pi,\pi]}(\lambda),\qquad
(\mathcal{F}\psi^{LP})(\lambda)=1_{[-2\pi,-\pi)\cup(\pi,2\pi]}(\lambda),
\]
condition \eqref{dyadic} is satisfied. Moreover, for any $j,j'\geq0$ and any $u,v\in\mathbb{R}$, the Parseval identity yields
\begin{equation}\label{decorrelate}
E\left[\tilde{\xi}_j(u)\tilde{\xi}_{j'}(v)\right]
=\int_{-\infty}^\infty\psi^{LP}_j(s-u)\psi^{LP}_{j'}(s-v)ds
=\frac{1}{2^{\frac{j+j'}{2}+1}\pi}\int_{\Lambda_j\cap\Lambda_{j'}} e^{\sqrt{-1}(u-v)\lambda}d\lambda,
\end{equation}
where $\Lambda_j=[-2^{j+1}\pi,-2^j\pi)\cup(2^j\pi,2^{j+1}\pi]$ for $j\in\mathbb{Z}$. Since $\Lambda_j\cap\Lambda_{j'}=\emptyset$ if $j\neq j'$, we have the independence of Brownian components across different levels because they are Gaussian.

Now, analogously to the first approach, we measure the lead-lag effect at the level $j$ by assessing the cross-covariance function between the processes $\tilde{\xi}_j^1(\cdot)$ and $\tilde{\xi}_j^2(\cdot)$. Here, we note that $\tilde{\xi}_j^1$ and $\tilde{\xi}_j^2$ are jointly stationary if the process $B_t=(B^1_t,B^2_t)$ is of stationary increments, i.e.~$E\left[\left(B^1_{t_1+h}-B^1_{s_1+h}\right)\left(B^2_{t_2+h}-B^2_{s_2+h}\right)\right]
=E\left[\left(B^1_{t_1}-B^1_{s_1}\right)\left(B^2_{t_2}-B^2_{s_2}\right)\right]$ 
for any $t_1,t_2,s_1,s_2,h\in\mathbb{R}$. As shown in the next subsection, the latter assumption is not restrictive. Consequently, we propose assessing the cross-covariance function:
\[
\rho_j(\theta)=\covariance\left[\tilde{\xi}_j^1(u),\tilde{\xi}_j^2(u+\theta)\right],\qquad\theta\in\mathbb{R},
\]
provided that {$\tilde{\xi}_j^1$ and $\tilde{\xi}_j^2$ are jointly (weakly) stationary}. 

\begin{rmk}
The ``coefficients'' $(f*\underline{\psi_j})_{j=0}^\infty$ in decomposition \eqref{LP} are called the \textit{(translation-invariant) dyadic wavelet transform} of $f$ in the wavelet literature (cf.~Section 5.2 of \cite{Mallat2009}). 
In conjunction with this terminology, the processes in \eqref{dyadic} might be called the dyadic wavelet transform of $dB^\nu$. 
Applications of dyadic wavelet transform based decomposition \eqref{LP} are found in pattern recognition and denoising with translation-invariant thresholding estimators; see Chapter 6 and Section 11.3.1 of \cite{Mallat2009}.  

\end{rmk}

\subsection{Model specification via cross-spectrum}

For ease of interpretation, it is desirable that $|\rho_j(\theta)|$ has the unique maximum at some $\theta_j\in\mathbb{R}$ for each $j$. So we presuppose such a situation and introduce the following specification into our model. 
We first note that \eqref{decorrelate} implies that $\tilde{\xi}^\nu_j$ has the spectral density $2^{-j}1_{\Lambda_j}$, hence its spectrum is concentrated on $\Lambda_j$. We also remark that if $W_t=(W^1_t,W^2_t)$ ($t\in\mathbb{R}$) is a bivariate two-sided Brownian motion with correlation parameter $R$, then for any $\theta\in\mathbb{R}$ the process $(W^1_t,W^2_{t-\theta})$ ($t\in\mathbb{R}$) is of stationary increments and has the cross-spectral density $Re^{-\sqrt{-1}\theta\lambda}$, $\lambda\in\mathbb{R}$. Motivated by these facts, let us suppose that the bivariate process $B_t=(B^1_t,B^2_t)$ is of stationary increments and its cross-spectral density is of the form $R_je^{-\sqrt{-1}\theta_j\lambda}$ with $R_j\in[-1,1]$ and $\theta_j\in\mathbb{R}$ for $\lambda\in\Lambda_j$, $j=0,1,\dots$. Namely, the function $f:\mathbb{R}\to\mathbb{C}$, which is defined by
\[
f(\lambda)=\sum_{j=0}^\infty R_je^{-\sqrt{-1}\theta_j\lambda}1_{\Lambda_j}(\lambda),\qquad\lambda\in\mathbb{R},
\]
satisfies
\begin{equation}\label{s-csd}
E\left[\left(\int_{-\infty}^\infty u(s)dB_s^1\right)\overline{\left(\int_{-\infty}^\infty v(s)dB^2_s\right)}\right]=\frac{1}{2\pi}\int_{-\infty}^\infty(\mathcal{F}u)(\lambda)\overline{(\mathcal{F}v)(\lambda)}f(\lambda)d\lambda
\end{equation}
for any $u,v\in L^2(\mathbb{R})$. Therefore, noting that
\begin{equation}\label{lp-fourier}
(\mathcal{F}^{-1}1_{\Lambda_0})(s)=\psi^{LP}(s)
\end{equation}
for $s\in\mathbb{R}$, in this situation we have for each $j$
\begin{align*}
\rho_j(\theta)
&=\frac{R_j}{2^{j+1}\pi}\int_{\Lambda_j}e^{\sqrt{-1}(\theta-\theta_j)\lambda}d\lambda
=R_j\psi^{LP}(2^j(\theta-\theta_j))
=R_j\frac{\sin[2^j\pi(\theta-\theta_j)]}{2^j\pi(\theta-\theta_j)}(2\cos[2^{j}\pi(\theta-\theta_j)]-1)
\end{align*}
by the Fourier inversion formula, hence $|\rho_j(\theta)|$ has the unique maximum $|R_j|$ at $\theta=\theta_j$ as long as $R_j\neq0$.

Now the question of matter is whether we can construct a bivariate process $B_t=(B^1_t,B^2_t)$ having the pre-described properties such that both $B^1$ and $B^2$ are respectively one-dimensional standard Brownian motions. The following proposition gives an affirmative answer:
\begin{proposition}\label{characterization}
Suppose that a measurable function $f:\mathbb{R}\to\mathbb{C}$ satisfies
\begin{equation}\label{L_infty}
\|f\|_\infty\leq1
\end{equation}
and
\begin{equation}\label{hermite}
\overline{f(\lambda)}=f(-\lambda)\qquad\text{for almost all }\lambda\in\mathbb{R}.
\end{equation}
Then there is a bivariate Gaussian process $B_t=(B^1_t,B^2_t)$ ($t\in\mathbb{R}$) with stationary increments such that
\begin{enumerate}[label={\normalfont(\roman*)}]

\item both $B^1$ and $B^2$ are two-sided Brownian motions,

\item $f$ is the cross-spectral density of $B$, i.e.~$f$ satisfies \eqref{s-csd} for any $u,v\in L^2(\mathbb{R})$.

\end{enumerate}

Conversely, if a bivariate process $B_t=(B^1_t,B^2_t)$ ($t\in\mathbb{R}$) with stationary increments satisfies condition (i), there is a measurable function $f:\mathbb{R}\to\mathbb{C}$ satisfying \eqref{L_infty}--\eqref{hermite} and condition (ii). 
\end{proposition}
We prove this result in Section \ref{proof:characterization}.

In the next section we will consider an asymptotic theory when the unit length $\tau_J$ shrinks to zero, or equivalently, when $J$ tends to infinity, which is a standard approach for theoretical analyses of statistics for high-frequency data (cf.~\citet{AJ2014}). In such a situation it is convenient to relabel indices of the parameters $R_j$ and $\theta_j$ so that the finest resolution $\tau_J$ corresponds to the level $j=1$. 
{That is, $R_j$ and $\theta_j$ are renamed $R_{J-j+1}$ and $\theta_{J-j+1}$ respectively.} 
This suggests that we should model the cross-spectral density of $B_t=(B^1_t,B^2_t)$ as
\begin{equation}\label{asymptotic}
f_{J}(\lambda)=\sum_{j=0}^J R_{{J-j+1}}e^{-\sqrt{-1}\theta_{{J-j+1}}\lambda}1_{\Lambda_j}(\lambda)
=\sum_{j=1}^{J+1} R_{j}e^{-\sqrt{-1}\theta_{j}\lambda}1_{\Lambda_{{J-j+1}}}(\lambda).
\end{equation}
Here, we omit all the components finer than $\tau_J$ from the model because they are not identifiable.

\begin{rmk}
It is well-known that in frictionless markets the no arbitrage assumption forces price processes to follow a semimartingale. As a consequence, the lead-lag model considered in this paper necessarily allows an arbitrage opportunity unless there is indeed no lead-lag relationship. 
However, this is not the case if we take account of some market friction such as the discreteness of transaction times or transaction costs. In fact, in \cite{HK2017arb} it has been shown that the market with a constant risk-free asset and two risky assets whose log-price processes are respectively given by $B^1$ and $B^2$ defined in the above is free of arbitrage if we impose a minimal waiting time on subsequent transactions or take proportional transaction costs into consideration (see Section 3.2 of \cite{HK2017arb}). 
\end{rmk}
   
\subsection{Construction of estimators}

If continuous-time observation data of the process $B_t$ on the whole real line were available, we could reproduce $\tilde{\xi}^\nu_{{J-j+1}}(s)$'s via
\[
\tilde{\xi}^\nu_{{J-j+1}}(s)=\int_{-\infty}^\infty\psi^{LP}_{{J-j+1}}(s-u)dB^\nu_u,\qquad
j=1,\dots,J+1
\]
for $\nu=1,2$ and construct estimators for $\rho_{{J-j+1}}(\theta)$ as in Section \ref{section:levy}. Since we only have discrete observation data of $B_t$ on $[0,n\tau_J]$, one natural way is to approximate the above integral by discretization. This is however problematic because discretization of $\psi^{LP}$ is unstable due to oscillation and $\psi^{LP}$ is not compactly supported. For these reasons we adopt another approach that {uses} Daubechies' compactly supported wavelets. Daubechies' wavelets generate finite-length filters whose gain functions well approximate those of the Littlewood-Paley wavelets (cf.~\citet{Lai1995}).

Specifically, we denote by $(h_p)_{p=0}^{L-1}$ Daubechies' wavelet filter of length $L$. Its squared gain function $H_L(\lambda)=|\sum_{p=0}^{L-1}h_{p}e^{-\sqrt{-1}\lambda p}|^2$ is given by
\[
H_L(\lambda)=2\sin^L(\lambda/2)\sum_{p=0}^{L/2-1}\binom{L/2-1+p}{p}\cos^{2p}(\lambda/2)
\]
(note that $L$ is always an even integer). The associated scaling filter $(g_p)_{p=0}^{L-1}$ is determined by the quadrature mirror relationship:
\[
g_p=(-1)^{p+1}h_{L-p-1},\qquad p=0,1,\dots,L-1.
\] 
Hence its squared gain function $G_L(\lambda)=|\sum_{p=0}^{L-1}g_{p}e^{-\sqrt{-1}\lambda p}|^2$ satisfies $G_L(\lambda)=H_L(\lambda-\pi)$. 
{See Section 4.8 of \cite{PW2000} for more details.}

\begin{rmk}
We follow \citet{PW2000} and \citet{GSW2002} in using the notation that $(h_p)$ denotes the wavelet filter and $(g_p)$ denotes the scaling filter. 
Note that, especially in the area of signal processing, many authors adopt the reverse usage: $(h_p)$ denotes the scaling filter and $(g_p)$ denotes the wavelet filter. 
See \cite{Daubechies1992,Vidakovic1999,Mallat2009} for instance. 
We also remark that the squared gain function of Daubechies' wavelet filter of length $L$ is often defined as $H_L(\lambda)/2$ rather than $H_L(\lambda)$ in the literature. We adopt the present definition to obtain the identity \eqref{normalize}. 
\end{rmk}

Using the filters $(g_p)_{p=0}^{L-1}$ and $(h_p)_{p=0}^{L-1}$ as {respectively} low-pass and high-pass filters, we can construct scale-by-scale band-pass filters. We denote by $(h_{j,p})_{p=0}^{L_j-1}$ the level $j$ wavelet filters associated with the filters $(g_p)_{p=0}^{L-1}$ and $(h_p)_{p=0}^{L-1}$ for every $j$, where $L_j=(2^j-1)(L-1)+1$. 
{See Section 4.6 of \cite{PW2000} for the precise definition.} 
For our analysis the form of its squared gain function $H_{j,L}(\lambda)=|\sum_{p=0}^{L_j-1}h_{j,p}e^{-\sqrt{-1}\lambda p}|^2$ is important, which is given by
\begin{equation}\label{hj-formula}
H_{j,L}(\lambda)=H_L(2^{j-1}\lambda)\prod_{i=0}^{j-2}G_L(2^i\lambda),\qquad\lambda\in\mathbb{R}.
\end{equation}
We also remark that the filter $(h_{j,p})_{p=0}^{L_j-1}$ has unit energy
\begin{equation}\label{normalize}
\sum_{p=0}^{L_j-1}h_{j,p}^2=1.
\end{equation}

Now we {approximate} $\tilde{\xi}^\nu_{{J-j+1}}(\cdot)$ by  
\begin{equation}\label{modwt}
\mathcal{W}_{jk}^\nu=\sum_{p=0}^{L_j-1}h_{j,p}\left(B^\nu_{(k+1-p)\tau_J}-B^\nu_{(k-p)\tau_J}\right),\qquad k=L_j-1,\dots,n-1
\end{equation}
for each $\nu=1,2$. We remark that the transformation of $(B^\nu_{(k+1)\tau_J}-B^\nu_{k\tau_J})_{k=0}^{n-1}$ in \eqref{modwt} is the so-called \textit{maximal overlap discrete wavelet transform} (MODWT) up to multiplication of constants, and that the resulting $\mathcal{W}_{jk}^\nu$'s are the corresponding wavelet coefficients. 
See {Section 5 of \cite{PW2000}} and Section 4.5 of \cite{GSW2002} for details.
Then we define the cross-covariance estimators at the level $j$ by
\begin{equation}\label{ccf-estimator}
\widehat{\rho}_{{J-j+1}}(l\tau_J)
=\left\{
\begin{array}{ll}
\frac{\tau_J^{-1}}{n-l-L_j+1}\sum_{k=L_j-1}^{n-l-1}\mathcal{W}^1_{jk}\mathcal{W}^2_{jk+l}  & \text{if }l\geq0  \\
\frac{\tau_J^{-1}}{n+l-L_j+1}\sum_{k=L_j-1}^{n+l-1}\mathcal{W}^1_{jk-l}\mathcal{W}^2_{jk}  & \text{otherwise}.
\end{array}
\right.
\end{equation}
In the next section we will show that these are asymptotically sensible cross-covariance estimators while both $J$ and $L$ tend to infinity with appropriate rates.

\begin{rmk}
The above estimators are formally the same as the \textit{wavelet cross-covariance estimators} used in discrete-time modeling framework up to multiplication of constants (cf.~Section 7.4 of \cite{GSW2002}). Therefore, the results presented in the next section ensure the validity of using such estimators in the continuous-time modeling framework proposed in this paper. 
\end{rmk}

\section{Asymptotic theory}\label{theory}


This section presents an asymptotic theory for the estimators constructed in the previous section. We also generalize the setting therein. First we assume that $n\tau_J\to T$ as $J\to\infty$ for some $T\in(0,\infty)$. This means that we observe data on a \textit{fixed} interval (e.g.~one day). Next, $B_t=(B_t^1,B_t^2)$ ($t\in\mathbb{R}$) denotes a bivariate Gaussian process with stationary increments such that both $B^1$ and $B^2$ are respectively one-dimensional two-sided standard Brownian motions and its cross-spectral density is given by \eqref{asymptotic}. We denote by $(\Omega,\mathcal{F},P)$ the probability space where the process $B$ is defined. Now, we assume that for each $\nu=1,2$ we have a filtration $(\mathcal{F}^\nu_t)_{t\geq0}$ such that the process $(B^\nu_t)_{t\geq0}$ is an $(\mathcal{F}^\nu_t)$-Brownian motion. Then, the $\nu$-th log-price process $X^\nu=(X^\nu_t)_{t\geq0}$ is given by
\begin{equation}\label{model}
X^\nu_t=X^\nu_0+\int_0^t\sigma^\nu_tdB^\nu_t,\qquad t\geq0,
\end{equation}
where $\sigma^\nu_t$ is an $(\mathcal{F}^\nu_t)$-adapted c\`adl\`ag process (hence the above stochastic integral is well-defined in the usual sense).

We assume that the processes $X^1$ and $X^2$ are observed at time stamps of the form $k\tau_J$ for some $k=0,1,\dots,n$. Unlike the previous section, we {allow observation data at some time points to be missing}. For each $\nu=1,2$ we denote by $\delta_{\nu,k}$ the indicator variable which is unity when $X^\nu$ is not observed at the time $k\tau_J$ and zero otherwise. So the observation data for $X^\nu$ are given by $\{X^\nu_{k\tau_J}:\delta_{\nu,k}=0,k=0,1,\dots,n\}$. For simplicity we assume that the initial value $X^\nu_0$ is observed, i.e.~$\delta_{\nu,0}\equiv0$. To construct our estimators we create (pseudo) complete observation data at all the points $k\tau_J$'s by use of the previous-tick interpolation scheme. Then we construct our estimators based on this observation data as in the previous section. To derive the mathematical expression of the estimators constructed so, it is convenient to introduce \citet{LM1990}'s notation. For each $\nu=1,2$, setting
\[
\chi_{\nu,k}(0)=1-\delta_{\nu,k+1},\qquad
\chi_{\nu,k}(\alpha)=(1-\delta_{\nu,k+1})\prod_{l=1}^{\alpha}\delta_{\nu,k+1-l},\qquad\alpha=1,\dots,k
\]
for $k=0,1,\dots,n-1$, we can write the pseudo observed returns for $X^\nu$ based on the interpolated data as
\[
\Delta^o_{k}X^\nu=\sum_{\alpha=0}^{k}\chi_{\nu,k}(\alpha)\Delta_{k-\alpha}X^\nu,\qquad k=0,1,\dots,n-1,
\]
where $\Delta_{k}X^\nu=X^\nu_{(k+1)\tau_J}-X^\nu_{k\tau_J}$. Therefore, we have $\mathcal{W}^\nu_{jk}=\sum_{p=0}^{L_j-1}h_{j,p}\Delta^o_{k-p}X^\nu$ for $j=0,1,\dots$ and $k=L_j-1,\dots,n-1$. The estimator $\widehat{\rho}_{{J-j+1}}(l\tau_J)$ is then constructed by \eqref{ccf-estimator}. Here, for the construction of asymptotic results in a general form we additionally consider a functional version of $\widehat{\rho}_{{J-j+1}}(l\tau_J)$. We define the process $(\widehat{\boldsymbol{\rho}}_{{J-j+1}}(l\tau_J)_t)_{t\geq0}$ by
\[
\widehat{\boldsymbol{\rho}}_{{J-j+1}}(l\tau_J)_t
=\left\{
\begin{array}{ll}
\frac{\tau_J^{-1}}{n-l-L_j+1}\sum_{k=L_j-1}^{\lfloor \tau_J^{-1}t\rfloor-l}\mathcal{W}^1_{jk}\mathcal{W}^2_{jk+l}  & \text{if }l\geq0,  \\
\frac{\tau_J^{-1}}{n+l-L_j+1}\sum_{k=L_j-1}^{\lfloor \tau_J^{-1}t\rfloor+l}\mathcal{W}^1_{jk-l}\mathcal{W}^2_{jk}  & \text{otherwise}
\end{array}
\right.
\]
for $t\geq0$.\footnote{We set $\sum_{k=p}^q\equiv0$ if $p>q$ by convention.} Since we have $\widehat{\rho}_{{J-j+1}}(l\tau_J)=\widehat{\boldsymbol{\rho}}_{{J-j+1}}(l\tau_J)_{(n-1)\tau_J}$, it is enough to investigate the asymptotic properties of $(\widehat{\boldsymbol{\rho}}_{{J-j+1}}(l\tau_j)_t)_{t\geq0}$.
 
Regarding the mechanism of missing observations, we focus on the following simple situation as in \cite{LM1990} (known as \textit{missing completely at random}): 
\begin{enumerate}[label=(\roman*)]

\item The observation for $X^\nu$ can be missing at each $i\tau_J$ with probability $\pi_\nu$ for $\nu=1,2$,

\item missing observations occur independently.

\end{enumerate}
Under this situation, for each $\nu=1,2$ $(\delta_{\nu,k})_{k=1}^\infty$ is a sequence of i.i.d.~Bernoulli variables with provabilities $\pi_\nu$ and $1-\pi_\nu$ of taking {values} 1 and 0. Also, $(\delta_{1,k})_{k=1}^\infty$ and $(\delta_{2,k})_{k=1}^\infty$ are mutually independent.

To avoid boundary issues, we assume that the true lag parameters $\theta_j\in(-\delta,\delta)$ for all $j$ with some constant $\delta\in(0,T)$, and evaluate the cross-covariance function {on the finite grid} $\mathcal{G}_J=\{l\tau_J:l\in\mathbb{Z},|l\tau_J|<\delta\}.$

For processes $(Y^J_t)_{t\geq0}$ ($J=1,2,\dots$) and a process $(Y_t)_{t\geq0}$, we write $Y^J_t\xrightarrow{ucp}Y_t$ to express $\sup_{0\leq s\leq t}|Y^J_s-Y_s|\to^p0$ as $J\to\infty$ for any $t>0$.
\begin{theorem}\label{theorem:main}
Let $j\in\mathbb{N}$ be fixed. Suppose that {$L=L^{(J)}$ depends on $J$ so that} $L\to\infty$ and $L^2\tau_J\to0$ as $J\to\infty$. Let $(\vartheta_J)$ be a sequence of real numbers such that $\vartheta_J\in\mathcal{G}_J$ for every $J$. 
\begin{enumerate}[label={\normalfont(\alph*)}]

\item If $L^{-\frac{1}{2}}\tau_J^{-1}(\vartheta_J-\theta_{j})\to\infty$ as $J\to\infty$, then $\widehat{\boldsymbol{\rho}}_{{J-j+1}}(\vartheta_J)_t\xrightarrow{ucp}0$ as $J\to\infty$.

\item If $\tau_J^{-1}(\vartheta_J-\theta_j)\to b$ as $J\to\infty$ for some $b\in\mathbb{R}$, then $\widehat{\boldsymbol{\rho}}_{{J-j+1}}(\vartheta_J)_t\xrightarrow{ucp}2^j\Sigma_t(\theta_j)R_{j}\int_{\Lambda_{-j}}D(\lambda)\Pi(\lambda)e^{\sqrt{-1}b\lambda}d\lambda$ as $J\to\infty$, where
\[
D(\lambda)=\frac{1}{2\pi}\left|\frac{e^{-\sqrt{-1}\lambda}-1}{\lambda}\right|^2,\qquad
\Pi(\lambda)=\frac{(1-\pi_1)(1-\pi_2)}{(1-\pi_1e^{\sqrt{-1}\lambda})(1-\pi_2e^{-\sqrt{-1}\lambda})}
\]
and
\[
\Sigma_t(\theta)
=\left\{
\begin{array}{ll}
\frac{1}{T-\theta}\int_0^{(t-\theta)_+}\sigma^1_s\sigma^2_{s+\theta}ds  & \text{if }\theta\geq0,  \\
\frac{1}{T+\theta}\int_0^{(t+\theta)_+}\sigma^1_{s-\theta}\sigma^2_{s}ds  & \text{otherwise}.
\end{array}
\right.
\]

\end{enumerate}
\end{theorem}
The proof of Theorem \ref{theorem:main} is given in Section \ref{proof:theorem1}.
\begin{rmk}
The first part of Theorem \ref{theorem:main} claims that our cross-covariance estimator $\widehat{\rho}_{{J-j+1}}(\theta)$ is close to zero if $\theta$ is sufficiently far from the true lag parameter $\theta_j$. The second part of the theorem claims that our cross-covariance estimator $\widehat{\rho}_{{J-j+1}}(\theta)$ tends to the quantity $R_j$ multiplied by some (random) constant. This constant consists of four sources: $\Sigma_t(\theta_j)$ comes from the presence of volatility. $D(\lambda)$ represents a discretization error caused by $dX^\nu$ being replaced by $\Delta_{k}X^\nu$'s. $\Pi(\lambda)$ is caused by previous-tick interpolation. $e^{\sqrt{-1}b\lambda}$ comes from the bias due to the discrepancy between $\theta$ and $\theta_j$.    
\end{rmk}
{
\begin{rmk}
In Theorem \ref{theorem:main}, $L$ can go to $\infty$ arbitrarily slowly because we are concerned only with the law of large numbers for the processes $(\widehat{\boldsymbol{\rho}}_{{J-j+1}}(\vartheta_J)_t)_{t\geq0}$. However, the divergence rate of $L$ affects the convergence rate of $(\widehat{\boldsymbol{\rho}}_{{J-j+1}}(\vartheta_J)_t)_{t\geq0}$ through the convergence rate of $H_L(\lambda)$ and the divergence rate of $H_L'(\lambda)$ as $L\to\infty$, and a too small value of $L$ would cause a slower convergence rate.
\end{rmk}
}
Theorem \ref{theorem:main} suggests that we could estimate the lag parameter $\theta_j$ by maximizing $|\widehat{\rho}_{{J-j+1}}(\theta)|$ over the finite grid $\mathcal{G}_J$ as in \citet{HRY2013}. Consequently, we choose the random variable $\widehat{\theta}_j\in\mathcal{G}_J$ such that 
\[
\left|\widehat{\rho}_{{J-j+1}}(\widehat{\theta}_j)\right|=\max_{\theta\in\mathcal{G}_J}\left|\widehat{\rho}_{{J-j+1}}(\theta)\right|
\]
as an estimator for $\theta_j$.
\begin{theorem}\label{HRY}
Let $j\in\mathbb{N}$ be fixed. Suppose that {$L=L^{(J)}$ depends on $J$ so that} $L\to\infty$ and $L^2\tau_J^\kappa\to0$ as $J\to\infty$ for some $\kappa\in(0,1)$. 
Suppose also that both $\sigma^1$ and $\sigma^2$ almost surely have $\gamma$-H\"older continuous paths for some $\gamma>0$. 
Then, if a sequence $v_J$ of positive numbers satisfies $L^{-\frac{1}{2}}\tau_J^{-1}v_J\to\infty$ as $J\to\infty$, then $v_J^{-1}(\widehat{\theta}_j-\theta_j)\to^p0$ 
as $J\to\infty$, provided that $R_j\neq0$ and $\Sigma_T(\theta_j)\neq0$ a.s. In particular, {we can take $v_J\equiv1$ and thus} we have $\widehat{\theta}_j\to^p\theta_j$ as $J\to\infty$.
\end{theorem}
We prove Theorem \ref{HRY} in Section \ref{proof:HRY}.

\begin{rmk}
The H\"older continuity assumption on the paths of the volatilities $\sigma^1$ and $\sigma^2$ is satisfied by a number of stochastic volatility models. In fact, it is satisfied if $\sigma^\nu$ is a continuous It\^o semimartingale (with respect to the filtration $(\mathcal{F}^\nu_t)$) for every $\nu=1,2$. It is worth mentioning that the assumption is also satisfied by the so-called rough fractional stochastic volatility models introduced in \citet{GJR2014} which have pointed out the practicality of such models.    
\end{rmk}

\begin{rmk}
Theorem \ref{HRY} is a counterpart of Theorem 1 from \citet{HRY2013} in our framework. Our results are applicable for separating multiple lead-lag effects on a scale-by-scale basis. On the other hand, the convergence rate is slower than the estimator of  \citet{HRY2013} by the factor $\sqrt{L}$, which might be regarded as a cost to separate multiple lead-lag effects.    
\end{rmk}

{
\begin{rmk}
In fact, to prove the stated results, we use only the following two properties of Daubechies' filter: 
\begin{enumerate}

\item The squared gain function of the filter ``well approximates'' $2\cdot1_{[-\pi,\frac{\pi}{2})\cup(\frac{\pi}{2},\pi]}$. 

\item The derivative of the squared gain function ``moderately'' diverges (note that it must diverge because of the first property). 

\end{enumerate}
Therefore, we conjecture that it would be possible to work with other filters once we precisely formulate the meaning of ``well approximates'' and ``moderately'' in the above properties. We leave this task to future work. 
\end{rmk}
}

\section{Simulation study}\label{simulation}

In this section we implement a Monte Carlo experiment to evaluate finite sample performance of our scale-by-scale lead-lag parameter estimators $\widehat{\theta}_j$'s defined in the previous section. We set $J=13$ and $n=15,000$. As the search grid $\mathcal{G}_J$, we use $\mathcal{G}_J=\{l\tau_J: l\in\mathbb{Z}, |l|\leq60\}.$


We simulate model \eqref{model} with constant volatilities $\sigma^\nu\equiv1$ for $\nu=1,2$. The parameters for the spectral density \eqref{asymptotic} are chosen as in Table \ref{parameters}. To simulate the process $B$ at the time points $k\tau_J$, $k=0,1,\dots,n$, it is enough to generate bivariate variables $\Delta_kB$, $k=0,1,\dots,n$. Since they are stationary and Gaussian, we can use the multivariate version of the circulant embedding method from \cite{CW1999} once we compute the cross-covariance function $E[\Delta_kB^1\Delta_{k+l}B^2]$, $l\in\mathbb{Z}$. It can be computed by using \eqref{lp-fourier} and the Fourier inversion formula as
\[
E\left[\Delta_k B^1\Delta_{k+l} B^2\right]=\sum_{j=0}^JR_{{J-j+1}}\int_0^{\tau_J}\int_0^{\tau_J}\psi^{LP}(2^j(u-v+l\tau_J-\theta_{{J-j+1}}))dudv,\quad l\in\mathbb{R},
\]
hence we approximate it by
\[
\tau_J^2\sum_{j=0}^JR_{{J-j+1}}\psi^{LP}(2^j(l\tau_J-\theta_{{J-j+1}})).
\]

\begin{table}[ht]
\caption{Parameters for the spectral density \eqref{asymptotic}}
\label{parameters}
\begin{center}
\begin{tabular}{l|*{10}{c}}\hline
$j$ & 1 & 2 & 3 & 4 & 5 & 6 & 7 & 8 & 9--14 \\ \hdashline
$R_j$ & 0.3 & 0.5 & 0.7 & 0.5 & 0.5 & 0.5 & 0.5 & 0.5 & 0 \\
$\theta_j/\tau_J$  & $-1$ & $-1$ & $-2$ & $-2$ & $-3$ & $-5$ & $-7$ & $-10$ & 0 \\ \hline
\end{tabular}
\end{center}
\end{table}%

Regarding the parameters $\pi^1$ and $\pi^2$ which controls the probabilities of missing observations, we consider two situations where $\pi_1=\pi_2=0$ and $\pi_1=\pi_2=0.5$. In the first situation no missing observation occurs, so the observation times are equidistant and synchronous. In the meantime, in the second situation the observation times are non-equidistant and asynchronous.

As the wavelet filter $(h_p)_{p=0}^{L-1}$ to construct our estimator, we examine the following three choices of Daubechies' wavelets:
\begin{enumerate}[font=\bfseries,align=left,leftmargin=*,labelindent=\parindent,widest=LA(20)]

\item[Haar] Haar wavelets ($L=2$),

\item[LA(8)] Least asymmetric wavelet with length $L=8$,

\item[LA(20)] Least asymmetric wavelet with length $L=20$.  

\end{enumerate}
See Chapter 8 of \cite{Daubechies1992} for details on the least asymmetric wavelets. 
{For comparison, we also compute the estimator of \citet{HRY2013} (denoted by HRY).} 

We run 1,000 Monte Carlo iterations for each experiment. We report the sample median and median absolute deviation (MAD) of the estimates for each experiment in Tables \ref{table:results}. As the table reveals, at the finest resolution levels $j=1,2$ all the estimates are very precise, while at coarser resolution levels $j\geq3$ the LA(8) and LA(20) based estimators tend to perform better than the Haar based estimators. This is not surprising because the consistency of our estimator is ensured in the asymptotics as $L$ tends to infinity. The LA(20) based estimator shows an excellent performance at moderate resolution levels $j\leq 6$ even in the presence of missing observations. At the coarsest resolution levels $j=7,8$, the precisions of all the estimators fall. This would be due to the following reason: According to the proofs of Theorems \ref{theorem:main}--\ref{HRY}, the convergence rate of our estimator $\widehat{\theta}_j$ is proportional to the square root of $L_j=(2^j-1)(L-1)+1$, hence {it becomes slower} as $j$ increases. 

\begin{table}[ht]
\caption{Simulation results}
\label{table:results}
\begin{center}
\begin{tabular}{l|*{8}{C}}
  \hline
$j$ & 1 & 2 & 3 & 4 & 5 & 6 & 7 & 8 \\ 
  \hline
True & -1 & -1 & -2 & -2 & -3 & -5 & -7 & -10 \\ 
   &  \multicolumn{8}{c}{$\pi_1=\pi_2=1$}  \\ 
   {HRY} & \multicolumn{8}{C}{{-1~(0)}} \\
  Haar & -1 & -1 & -1 & -1 & -1 & -1 & -2 & -2 \\ 
   & (0) & (0) & (0) & (0) & (0) & (0) & (0) & (1) \\ 
  LA(8) & -1 & -1 & -2 & -2 & -3 & -4 & -6 & -8 \\ 
   & (0) & (0) & (0) & (0) & (0) & (1) & (3) & (10) \\ 
  LA(20) & -1 & -1 & -2 & -2 & -3 & -5 & -6 & -9 \\ 
   & (0) & (0) & (0) & (0) & (0) & (1) & (4) & (15) \\ 
   &  \multicolumn{8}{c}{$\pi_1=\pi_2=0.5$}  \\ 
   {HRY} & \multicolumn{8}{C}{{-1~(0)}} \\
  Haar & -1 & -1 & -2 & -2 & -2 & -2 & -2 & -2 \\ 
   & (0) & (0) & (0) & (0) & (0) & (0) & (0) & (0) \\ 
  LA(8) & -1 & -1 & -2 & -2 & -3 & -4 & -6 & -8 \\ 
   & (0) & (0) & (0) & (0) & (0) & (1) & (3) & (10) \\ 
  LA(20) & -1 & -1 & -2 & -2 & -3 & -5 & -6 & -9 \\ 
   & (0) & (0) & (0) & (0) & (0) & (1) & (4) & (15) \\ 
   \hline
\end{tabular}\\ \vspace{2mm}

\parbox{9cm}{\small This table reports the median and the median absolute deviation (in parentheses) of the estimates.} 
\end{center}
\end{table}

\section{Empirical illustration}\label{empirical}

In this section we apply our new method to evaluating lead-lag effects on a scale-by-scale basis in real financial data. Specifically, we analyze the lead-lag relationships between the log-prices of the S\&P 500 index and the E-mini S\&P500 futures in April 2016, containing 21 trading days. We have obtained our data set {from the Bloomberg database}. The price data of the S\&P 500 index are recorded with the second-by-second basis from 9:30 am EDT to 16:00 EDT each trading day. We use the transaction data of the E-mini S\&P500 futures recorded between 9:30 am EDT and 16:00 EDT each trading day with the accuracy in the timestamp values being one second.

Before presenting the result, we remark that there are {many researchers} examining the lead-lag effect between the S\&P 500 index and index futures: See \cite{KKK1987,SW1990,deJN1997} for example. These studies have reported that the futures lead the index. 

We regard the log-price process of the S\&P 500 index as $X^1$ and the log-price process of the E-mini S\&P 500 futures as $X^2$. We set 
\[
\mathcal{G}_J=\{-300s,-299s,\dots,-1s,0s,1s,\dots,299s,300s\}
\]
as the search grid. We use LA(20) (Daubechies' least asymmetric wavelet filter with length 20) as the wavelet filter $(h_p)_{p=0}^{L-1}$.

Table \ref{empirical1} reports the estimated values of $\widehat{\theta}_j$ for $j=1,\dots,8$ on April 1, 2016. 
{For comparison, we report the results of \citet{HRY2013}'s estimator as well (referred to as HRY).} 
We also depict the function $\widehat{\rho}_{{J-j+1}}(\theta)$ evaluated for $\theta\in\{-30s,-29s,\dots,-1s,0s,1s,\dots,29s,30s\}$ in Figure \ref{figure:wccf}. 
The table shows that all the estimated lags are negative, which indicates that $X^2$ (futures) lead $X^1$ (index). This is consistent with the preceding studies. 
{We also observe that the absolute value of $\hat{\theta}_j$ increases as $j$ increases. This phenomenon might be caused by the difference between reaction speeds of market participants: In our model, the larger $j$, the coarser time scale becomes. We reasonably expect that market participants at coarser time scales would react at lower frequencies, so the absolute value of $\theta_j$ would be larger for larger $j$.} 

In Figure \ref{theta-all} we depict the time series of $\widehat{\theta}_j$'s evaluated every trading day\footnote{{We also computed \citet{HRY2013}'s estimator every trading day, which always estimated the lag $-1$ across all the days.}}. We find that the estimated values of $\widehat{\theta}_j$'s are quite stable in this period, especially at finer resolutions $j\leq 6$. This suggests that there might be a stable multi-scale structure in lead-lag effects between the S\&P 500 index and the E-mini S\&P500 futures.  

\begin{table}[ht]
\caption{Estimated values of $\widehat{\theta}_j$ for April 1, 2016 (in seconds)}
\label{empirical1}
\begin{center}
\begin{tabular}{c|cccccccc||cc}\hline
$j$ & 1 & 2 & 3 & 4 & 5 & 6 & 7 & 8 & {HRY} \\ \hdashline
$\widehat{\theta}_j$ & $-1$ & $-1$ & $-2$ & $-2$ & $-3$ & $-5$ & $-7$ & $-10$ & {$-1$} \\ \hline
\end{tabular}
\end{center}
\end{table}%

\begin{figure}[ht]
\centering
\caption{The function $\widehat{\rho}_{{J-j+1}}(\theta)$ for April 1, 2016}
\label{figure:wccf}
\includegraphics[scale=0.8]{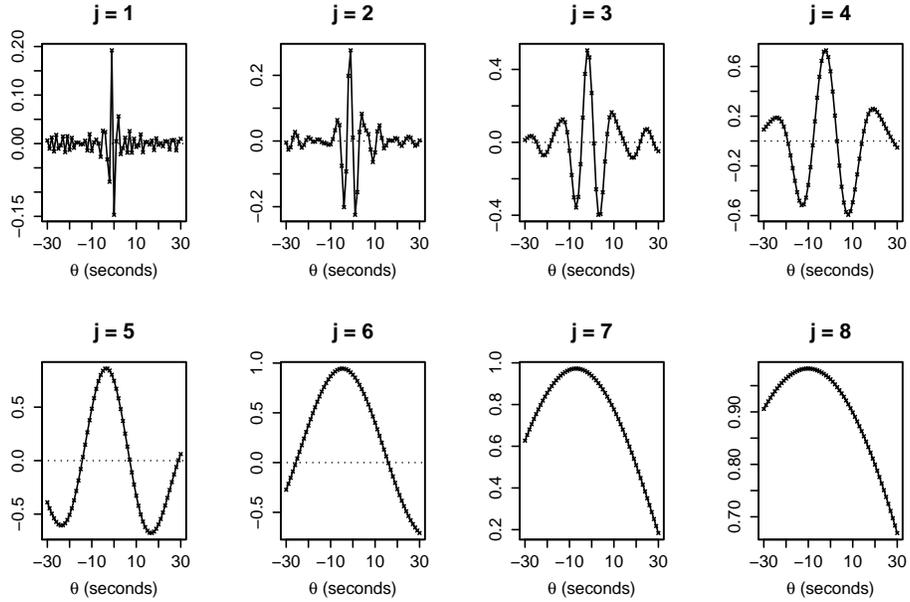}

\parbox{12.3cm}{\small
The values are divided by $\frac{\tau_J^{-1}}{n-L_j+1}\sqrt{(\sum_{k=L_j-1}^{n-1}(\mathcal{W}_{jk}^1)^2)(\sum_{k=L_j-1}^{n-1}(\mathcal{W}_{jk}^2)^2)}$ for each $j$.
}
\end{figure}


\begin{figure}[ht]
\centering
\caption{The time series of the estimates $\hat{\theta}_{j}$ in April, 2016 ($1\leq j\leq8$)}
\label{theta-all}
\includegraphics[scale=0.8]{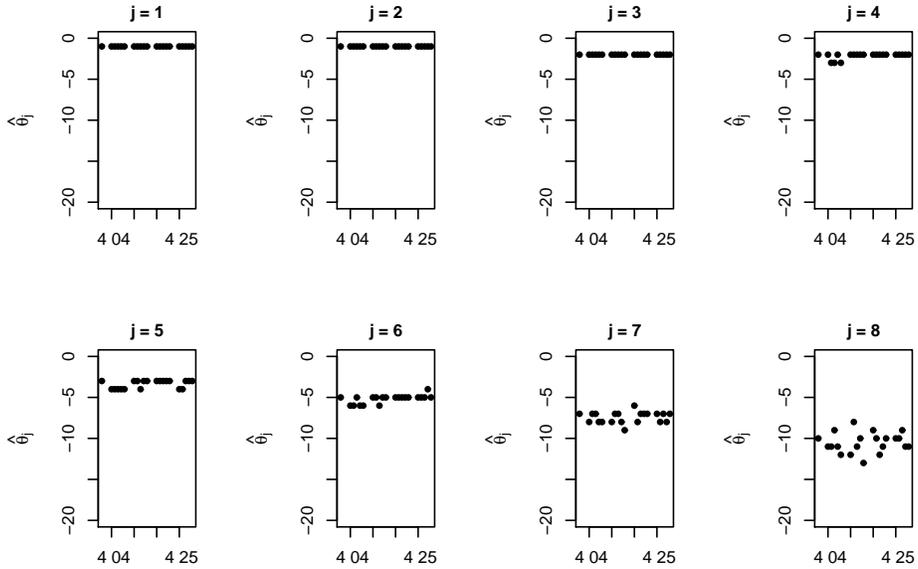}

\parbox{10cm}{\small
The horizontal axis is labeled by days. The vertical axis is in seconds.
}
\end{figure}

\clearpage

\section{Summary}

In this paper, we have proposed a new framework to model potential multi-scale structures of lead-lag relationships between financial assets. Our framework has accommodated traditional wavelet methods for analyzing lead-lag relationships to continuous-time modeling by establishing an explicit connection between wavelet and the L\'evy-Ciesielski construction of Brownian motion from a statistical viewpoint. We have also developed a statistical methodology to estimate lead-lag times on a scale-by-scale basis in the proposed framework. An associated asymptotic theory has been shown in order to ensure the validity of our methodology. To complement the theory, we have conducted a simulation study which demonstrates the finite sample performance of our asymptotic theory. We have reported an empirical application as well to illustrate how the proposed framework performs in practice.  

A drawback of the presented estimation method is that it requires us to interpolate data onto the grid with the finest observable resolution. In some cases this is {computationally} challenging. For example, if the finest observable resolution is one micro-second, then the method requires us to store one million observations per one second. A solution to this issue is currently under investigation and will be presented in future work. 

\section{Proofs}\label{proofs}

\subsection{Proof of Proposition \ref{prop:dyadic}}\label{proof:dyadic}

{We first note that $1_{[0,t]}*\underline{g}\in L^2(\mathbb{R})$ for every $g=\tilde{\phi},\tilde{\psi}_0,\tilde{\psi}_1,\dots$ by Young's convolution theorem.} 
Next, by Theorem 5.11 from \cite{Mallat2009} we have
$
1_{[0,t]}=(1_{[0,t]}*\underline{\tilde{\phi}})*\tilde{\phi}+\sum_{j=0}^\infty 2^j(1_{[0,t]}*\underline{\tilde{\psi}_j})*\tilde{\psi}_j
$ in $L^2(\mathbb{R}).$
Theorem 5.11 from \cite{Mallat2009} also implies that $(1_{[0,t]}*\underline{\tilde{\phi}})*\tilde{\phi},(1_{[0,t]}*\underline{\tilde{\psi}_j})*\tilde{\psi}_j\in L^2(\mathbb{R})$. Therefore, we obtain
\begin{equation*}
B_t=\int_{-\infty}^\infty (1_{[0,t]}*\underline{\tilde{\phi}})*\tilde{\phi}(s)dB_s+\sum_{j=0}^\infty 2^j\int_{-\infty}^\infty (1_{[0,t]}*\underline{\tilde{\psi}_j})*\tilde{\psi}_j(s)dB_s
\end{equation*}
in $L^2(P)$ {by the $L^2$-continuity of It\^o integrals}. 
{
Hence the proof is completed once we show that
\[
\int_{-\infty}^\infty (1_{[0,t]}*\underline{g})*g(s)dB_s
=\lim_{A\to\infty}\int_0^t\left\{\int_{-A}^A g(v-u)\left(\int_{-\infty}^\infty g(s-u)dB_s\right)du\right\}dv\quad
\text{in }L^2(P)
\]
for $g=\tilde{\phi},\tilde{\psi}_0,\tilde{\psi}_1,\dots$. Since we have
\[
\int_{-\infty}^\infty (1_{[0,t]}*\underline{g})*g(s)dB_s
=\lim_{A\to\infty}\int_{-\infty}^\infty\left\{\int_{-A}^A(1_{[0,t]}*\underline{g})(u)g(s-u)du\right\}dB_s\quad
\text{in }L^2(P),
\]
it suffices to show that
\begin{equation}\label{fubini}
\int_{-\infty}^\infty\left\{\int_{-A}^A(1_{[0,t]}*\underline{g})(u)g(s-u)du\right\}dB_s
=\int_0^t\left\{\int_{-A}^A g(v-u)\left(\int_{-\infty}^\infty g(s-u)dB_s\right)du\right\}dv
\quad
\text{a.s.}
\end{equation}
for every $A>0$. Since we have
\begin{align*}
\int_{-\infty}^\infty\left\{\int_{-A}^A(1_{[0,t]}*\underline{g})(u)g(s-u)du\right\}^2ds
=\|g\|_{L^2(\mathbb{R})}^2\left\{\int_{-A}^A(1_{[0,t]}*\underline{g})(u)du\right\}^2<\infty,
\end{align*}
we obtain
\[
\int_{-\infty}^\infty\left\{\int_{-A}^A(1_{[0,t]}*\underline{g})(u)g(s-u)du\right\}dB_s
=\int_{-A}^A\left\{\int_{-\infty}^\infty(1_{[0,t]}*\underline{g})(u)g(s-u)dB_s\right\}du
\quad
\text{a.s.}
\]
by the stochastic Fubini theorem (e.g.~Theorem 65 from chapter IV of \cite{Protter2004}). Since we have
\[
\int_{-A}^A\left\{\int_{-\infty}^\infty(1_{[0,t]}*\underline{g})(u)g(s-u)dB_s\right\}du
=\int_0^t\left\{\int_{-A}^A g(v-u)\left(\int_{-\infty}^\infty g(s-u)dB_s\right)du\right\}dv
\quad
\text{a.s.}
\]
by the Fubini theorem, we obtain \eqref{fubini} and thus complete the proof.\hfill\qed
}

\subsection{Proof of Proposition \ref{characterization}}\label{proof:characterization}

We use some concepts on Schwartz's generalized functions and refer to Chapters 6--7 of \cite{Rudin1991} for details about them.

{Let us denote by $\mathfrak{S}$ the space of all (complex-valued, smooth) rapidly decreasing functions on $\mathbb{R}$}. We denote by $\mathfrak{S}^*$ the dual space of $\mathfrak{S}$, and for $F\in\mathfrak{S}^*$ $\mathcal{F}F$ denotes the Fourier transform of $F$.  
Define the function $\mathbf{f}:\mathfrak{S}\to\mathbb{C}$ by $\mathbf{f}(u)=\int_{-\infty}^\infty(\mathcal{F}^{-1}u)(\lambda)f(\lambda)d\lambda$ for $u\in\mathfrak{S}$, which can be defined thanks to \eqref{L_infty}. We have $\mathbf{f}\in\mathfrak{S}^*$ because $f\in L^\infty(\mathbb{R})$. 
Moreover, if $u\in\mathfrak{S}$ is real-valued, then $\mathbf{f}(u)\in\mathbb{R}$. In fact, we have
\begin{align*}
\overline{\mathbf{f}(u)}=\int_{-\infty}^\infty\overline{(\mathcal{F}^{-1}u)(\lambda)}\overline{f(\lambda)}d\lambda
=\int_{-\infty}^\infty(\mathcal{F}^{-1}u)(-\lambda)f(-\lambda)d\lambda
=\mathbf{f}(u)
\end{align*}
by \eqref{hermite}. Now, for any $u\in\mathfrak{S}$ we have $\mathcal{F}(\mathbf{f}*u)=(\mathcal{F}u)(\mathcal{F}\mathbf{f})=(\mathcal{F}u)f$ in $\mathfrak{S}^*$ by Theorem 7.19(c) from \cite{Rudin1991}, hence $\mathcal{F}(\mathbf{f}*u)\in L^2(\mathbb{R})$. Therefore, $\mathbf{f}*u\in L^2(\mathbb{R})$ and $\|\mathbf{f}*u\|_{L^2(\mathbb{R})}=(2\pi)^{-1}\|(\mathcal{F}u)f\|_{L^2(\mathbb{R})}\leq\|u\|_{L^2(\mathbb{R})}$ by the Parseval identity and \eqref{L_infty}. Hence, there is a (unique) continuous function $\mathbf{F}:L^2(\mathbb{R})\to L^2(\mathbb{R})$ such that $\mathbf{F}(u)=\mathbf{f}*u$ for any $u\in\mathfrak{S}$. By continuity $\mathbf{F}(u)$ is real-valued as long as $u\in L^2(\mathbb{R})$ is real-valued. Similarly, we define the function $\mathbf{g}:\mathfrak{S}\to\mathbb{C}$ by $\mathbf{g}(u)=\int_{-\infty}^\infty(\mathcal{F}^{-1}u)(\lambda)\sqrt{1-|f(\lambda)|^2}d\lambda$ for $u\in\mathfrak{S}$. Then, by an analogous argument to the above, $\mathbf{g}\in\mathfrak{S}^*$ and there is a (unique) continuous function $\mathbf{G}:L^2(\mathbb{R})\to L^2(\mathbb{R})$ such that $\mathbf{G}(u)=\mathbf{g}*u$ for any $u\in\mathfrak{S}$ and that $\mathbf{G}(u)$ is real-valued as long as so is $u\in L^2(\mathbb{R})$.

Now let $(W^1_t)_{t\in\mathbb{R}}$ and $(W^2_t)_{t\in\mathbb{R}}$ be two independent two-sided standard Brownian motions. Then we define the processes $(B^1_t)_{t\in\mathbb{R}}$ and $(B^2_t)_{t\in\mathbb{R}}$ by $B^1_t=W^1_t$ and
\[
B^2_t=
\left\{
\begin{array}{ll}
\int_{-\infty}^\infty\mathbf{F}(1_{(0,t]})(s)dW^1_s+\int_{-\infty}^\infty\mathbf{G}(1_{(0,t]})(s)dW^2_s  & \text{if }t\geq0,  \\
-\int_{-\infty}^\infty\mathbf{F}(1_{(t,0]})(s)dW^1_s-\int_{-\infty}^\infty\mathbf{G}(1_{(t,0]})(s)dW^2_s  & \text{otherwise}. 
\end{array}
\right.
\]
{$B^2$ is obviously real-valued. Moreover, noting that $\mathcal{F}\mathbf{F}(u)=(\mathcal{F}u)f$ and $\mathcal{F}\mathbf{G}(u)=(\mathcal{F}u)\sqrt{1-|f|^2}$ in $L^2(\mathbb{R})$ for any $u\in L^2(\mathbb{R})$, we can easily check that $B^2$ is a two-sided standard Brownian motion due to the Parseval identity and the Kolmogorov continuity criterion, where we take a continuous version of $B^2$ if necessary.} 
Hence $B^1$ and $B^2$ satisfy condition (i). Condition (ii) also follows from the Parseval identity. This especially implies that the bivariate process $B_t=(B_t^1,B_t^2)$ is of stationary increments. $B_t$ is evidently Gaussian by construction, hence we complete the proof of the first part of the proposition. 

Conversely, suppose that a bivariate process $B_t=(B^1_t,B^2_t)$ ($t\in\mathbb{R}$) with stationary increments satisfies condition (i). Then, by the spectral representation theorem for the structural function of a process with stationary increments (see e.g.~Theorem 4 from Chapter I, Section 5 of \cite{GS1969}) there is a function $F:\mathbb{R}\to\mathbb{C}$ of bounded variation such that
\begin{equation}\label{spr2}
E\left[\left(\int_{-\infty}^\infty u(s)dB_s^1\right)\overline{\left(\int_{-\infty}^\infty v(s)dB^2_s\right)}\right]=\frac{1}{2\pi}\int_{-\infty}^\infty(\mathcal{F}u)(\lambda)\overline{(\mathcal{F}v)(\lambda)}dF(\lambda)
\end{equation}
for any $u,v\in L^2(\mathbb{R})$ and 
{
the matrix
\[
\begin{pmatrix}
\lambda_2-\lambda_1 & F(\lambda_2)-F(\lambda_1) \\
\overline{F(\lambda_2)}-\overline{F(\lambda_1)} & \lambda_2-\lambda_1
\end{pmatrix}
\]
is positive semidefinite whenever $\lambda_1\leq\lambda_2$. Here, the latter property follows from the property (a) of the theorem from \cite{GS1969} and condition (i) (so the spectral distribution functions of the marginals of $B$ are the identity functions). 
In particular, the determinant of the above matrix is non-negative, hence we obtain 
\[
(\lambda_2-\lambda_1)^2-|F(\lambda_2)-F(\lambda_1)|^2\geq0,
\]
Thus we have
}
\begin{equation}\label{ac}
|F(\lambda_1)-F(\lambda_2)|\leq|\lambda_1-\lambda_2|
\end{equation}
for any $\lambda_1,\lambda_2\in\mathbb{R}$. \eqref{ac} implies that $F$ is absolutely continuous, so $F$ is differentiable almost everywhere on $\mathbb{R}$ and $f:=F'\in L^1_\text{loc}(\mathbb{R})$ by Theorem 7.20 from \cite{Rudin1987}. {Thus} $f$ satisfies \eqref{s-csd} due to \eqref{spr2}. Moreover, \eqref{L_infty} follows from \eqref{ac} and Theorem 1.40 from \cite{Rudin1987}. Finally, by \eqref{s-csd} we have
\begin{align*}
&\int_{-\infty}^\infty\left|(\mathcal{F}u)(\lambda)\right|^2\overline{f(\lambda)}d\lambda
=2\pi E\left[\overline{\left(\int_{-\infty}^\infty u(s)dB_s^1\right)}\left(\int_{-\infty}^\infty u(s)dB^2_s\right)\right]\\
&=2\pi E\left[\left(\int_{-\infty}^\infty \overline{u(s)}dB_s^1\right)\overline{\left(\int_{-\infty}^\infty \overline{u(s)}dB^2_s\right)}\right]
=\int_{-\infty}^\infty\left|(\mathcal{F}\overline{u})(\lambda)\right|^2f(\lambda)d\lambda\\
&=\int_{-\infty}^\infty\left|(\mathcal{F}u)(-\lambda)\right|^2f(\lambda)d\lambda
=\int_{-\infty}^\infty\left|(\mathcal{F}u)(\lambda)\right|^2f(-\lambda)d\lambda
\end{align*}
for any $u\in L^2(\mathbb{R})$. Therefore, for any bounded Borel set $A\subset\mathbb{R}$, we have $\int_A\overline{f(\lambda)}d\lambda=\int_Af(-\lambda)d\lambda$ 
by taking $u=\mathcal{F}^{-1}1_A$. Consequently, $f$ satisfies \eqref{hermite} due to Theorem 1.40 from \cite{Rudin1987}.\hfill$\Box$

\subsection{Proof of Theorem \ref{theorem:main}}\label{proof:theorem1}

Throughout the discussions, for sequences $(x_J)$ and $(y_J)$, $x_J\lesssim y_J$ means that there exists a constant $K\in[0,\infty)$ such that $x_J \leq K y_J$ for large $J$. Also, for $r>0$ $\|\cdot\|_r$ denotes the $L^r$-norm of random variables, i.e.~$\|\xi\|_r=(E[|\xi|^r])^{1/r}$ for a random variable $\xi$.

First we note that, without loss of generality, the volatility processes $\sigma^1$ and $\sigma^2$ may be assumed to be bounded by a standard localization argument presented in e.g.~{Section 1d from chapter I of \cite{JS2003}}. In fact, for each $K>0$ and each $\nu=1,2$, let us define $S^\nu_K=\inf\{t:|\sigma^\nu_t|>K\}$. We have $|\sigma^\nu_s|\leq K$ as long as $s<S^\nu_K$. Now define the process $\sigma^{\nu,(K)}_s$ by $\sigma^{\nu,(K)}_s=\sigma^\nu_s1_{\{s<S^\nu_K\}}+\sigma^\nu_{S^\nu_K-}1_{\{s\geq S^\nu_K\}}$ for $s\geq0$. $\sigma^{\nu,(K)}_s$ is obviously c\`adl\`ag and $(\mathcal{F}^\nu_t)$-adapted because $S^\nu_K$ is an $(\mathcal{F}^\nu_t)$-stopping time. So we can define the process $X^{\nu,(K)}_t$ by $X^{\nu,(K)}_t=\int_0^t\sigma^{\nu,(K)}_sdB^\nu_s$ for $t\geq0$. We associate $\widehat{\rho}_{{J-j+1}}^{(K)}(\theta)$ with $X^{1,(K)}$ and $X^{2,(K)}$. Now since we have $\{\widehat{\rho}_{{J-j+1}}^{(K)}(\theta)\neq\widehat{\rho}_{{J-j+1}}(\theta))\}\subset\{S^1_K\wedge S^2_K< t+1\}$ and
\begin{equation*}
\limsup_{K\to\infty}P\left(S^1_K\wedge S^2_K< t+1\right)\leq\limsup_{K\to\infty}P\left(\max_{\nu=1,2}\sup_{0\leq s\leq t+1}|\sigma^\nu_s|>K\right)=0
\end{equation*}
because both $\sigma^1$ and $\sigma^2$ are c\`adl\`ag, the results of Theorem \ref{theorem:main} hold true once they hold true for $\widehat{\rho}_{{J-j+1}}^{(K)}(\theta)$. Consequently, in the following we assume that both $\sigma^1$ and $\sigma^2$ are bounded.

{
We begin by proving the following auxiliary result. 
\begin{lemma}\label{lemma:ucp}
For each $n\in\mathbb{N}$, let $X^n=(X^n_t)_{t\geq0}$ be a c\`adl\`ag process. Also, let $X=(X_t)_{t\geq0}$ be a continuous process. Suppose that the following two conditions are satisfied:
\begin{enumerate}[label=(\roman*)]

\item\label{ass:tight} The sequence $(X^n)_{n\in\mathbb{N}}$ is tight for the Skorokhod topology.

\item\label{ass:margin} $X^n_t\to^pX_t$ as $n\to\infty$ for any $t\geq0$.

\end{enumerate}
Then we have $X^n\xrightarrow{ucp}X$ as $n\to\infty$. 
\end{lemma}

\begin{proof}
First, condition \ref{ass:margin} immediately yields $(X^n_{t_1}-X_{t_1},\dots,X^n_{t_k}-X_{t_k})\to^p(0,\dots,0)$ as $n\to\infty$ for any $k\in\mathbb{N}$ and $t_1,\dots,t_k\geq0$. Therefore, by Theorem 18.10(iii) of \cite{Vaart1998} the finite-dimensional distributions of the processes $X^n-X$ $(n=1,2,\dots)$ converge in law to those of the identically zero function on $[0,\infty)$ as $n\to\infty$. Since $X$ is continuous by assumption, the sequence $(X^n-X)_{n\in\mathbb{N}}$ is tight for the Skorokhod topology by condition \ref{ass:tight} and Corollary 3.{33} of chapter VI from \cite{JS2003}, hence we conclude that the processes $X^n-X$ converge in law to the identically zero function on $[0,\infty)$ for the Skorokhod topology. Finally, this implies that the processes $X^n-X$ converge in probability to the identically zero function on $[0,\infty)$ for the Skorokhod topology by Theorem 18.10(iii) of \cite{Vaart1998}, hence we obtain $X^n\xrightarrow{ucp}X$ as $n\to\infty$ by the continuity of $X$ and Proposition 1.17 of chapter VI from \cite{JS2003}. 
\end{proof}
}

We use the notation $\mathcal{L}_J=\{l\in\mathbb{Z}:l\tau_J\in\mathcal{G}_J\}$ and $\mathcal{L}_J^+=\{l\in\mathbb{Z}_+:l\tau_J\in\mathcal{G}_J\}.$ 
{We next prove the tightness of the target processes.}
\begin{lemma}\label{lemma:c-tight}
Under the assumptions of Theorem \ref{theorem:main}, {the sequence of processes $((\widehat{\boldsymbol{\rho}}_j(\vartheta_J)_t)_{t\geq0})_{J\in\mathbb{N}}$ is tight for the Skorokhod topology.}
\end{lemma}

\begin{proof}
Since $\vartheta_J\in\mathcal{G}_J$, it can be written as $\vartheta_J=l\tau_J$ for some $l=l_J\in\mathcal{L}_J$. For the simplicity of presentation we assume that $l_J\in\mathcal{L}_J^+$ for all $J$ (this assumption can easily be removed). 

According to Proposition 3.26 from Chapter VI of \cite{JS2003}, it suffices to prove
\begin{align}
&\lim_{A\to\infty}\limsup_{J\to\infty}P\left(\sup_{0\leq t\leq m}\left|\widehat{\boldsymbol{\rho}}_{{J-j+1}}(l\tau_J)_t\right|>A\right)=0,\label{c-aim1}\\
&\lim_{\eta\to0}\limsup_{J\to\infty}P\left(w_m(\widehat{\boldsymbol{\rho}}_{{J-j+1}}(l\tau_J),\eta)>\varepsilon\right)=0\label{c-aim2}
\end{align}
for any $m\in\mathbb{N}$ and any $\varepsilon>0$. Here for a function $g:[0,\infty)\to\mathbb{R}$ $w_m(g,\eta)$ denote the modulus of continuity of $g$ on $[0,m]$, i.e.~$w_m(g,\eta)=\sup\{|g(s)-s(t)|;s,t\in[0,m],|s-t|\leq\eta\}.$

First we show that there is a constant $C$ such that
\begin{equation}\label{c-fourth}
\left\|\mathcal{W}^\nu_{jk}\right\|_4\leq C\sqrt{\tau_J}
\end{equation}
for any $j,k$ and $\nu=1,2$. Since we can rewrite $\mathcal{W}^\nu_{jk}$ as 
$\mathcal{W}^\nu_{jk}=\sum_{\alpha=0}^{k}\chi_{\nu,k-p}(\alpha)\sum_{p=0}^{(L_j-1)\wedge(k-\alpha)}h_{j,p}\Delta_{k-p-\alpha}X^\nu,$ 
the Minkowski and the Burkholder-Davis-Gundy inequalities as well as \eqref{normalize} yield
\begin{align*}
\|\mathcal{W}^\nu_{jk}\|_4
&\leq\sum_{\alpha=0}^{k}\|\chi_{\nu,k-p}(\alpha)\|_4\left\|\sum_{p=0}^{(L_j-1)\wedge(k-\alpha)}h_{j,p}\Delta_{k-p-\alpha}X^\nu\right\|_4
\lesssim\sum_{\alpha=0}^{k}\pi_\nu^{\alpha/4}\left\{\left(\sum_{p=0}^{L_j-1}h^2_{j,p}\right)^2\cdot{\tau_J^2}\right\}^{1/4}
\lesssim\sqrt{\tau_J},
\end{align*}
hence \eqref{c-fourth} holds true.

Next, \eqref{c-fourth} and the Schwarz inequality imply that 
$E\left[\sup_{0\leq t\leq m}\left|\widehat{\boldsymbol{\rho}}_{{J-j+1}}(l\tau_J)_t\right|\right]\lesssim
\sum_{k=L_j-1}^{\lfloor \tau_J^{-1}m\rfloor-l}E\left[\left|\mathcal{W}^1_{jk}\mathcal{W}^2_{jk+l}\right|\right]\lesssim m,$ 
hence \eqref{c-aim1} holds true by the Markov inequality.

Finally, for $0\leq s\leq t\leq m$, the Schwarz inequality yields
\begin{align*}
\left|\widehat{\boldsymbol{\rho}}_{{J-j+1}}(l\tau_J)_t-\widehat{\boldsymbol{\rho}}_{{J-j+1}}(l\tau_J)_s\right|
\lesssim\left(\lfloor\tau_J^{-1}t\rfloor-\lfloor\tau_J^{-1}s\rfloor\right)\sum_{k=L_j-1}^{\lfloor\tau_J^{-1}m\rfloor-l}\left|\mathcal{W}^1_{jk}\mathcal{W}^2_{jk+l}\right|^2,
\end{align*}
hence we have
$
P\left(w_m\left(\widehat{\boldsymbol{\rho}}_{{J-j+1}}(l\tau_J),\eta\right)>\varepsilon\right)
\lesssim\varepsilon^{-1}(\eta+\tau_J)m
$
by the Markov and Schwarz inequalities as well as \eqref{c-fourth}. This yields \eqref{c-aim2}. This completes the proof.
\end{proof}
Next we assess the errors induced by interpolation. 
{
\begin{lemma}\label{lemma:lm-term}
Under the assumptions of Theorem \ref{theorem:main}, we have
\[
\left\|\widehat{\boldsymbol{\rho}}_{{J-j+1}}(l\tau_J)_{t}-E\left[\widehat{\boldsymbol{\rho}}_{{J-j+1}}(l\tau_J)_{t}|X\right]\right\|_2
\leq \frac{C\sqrt{t\tau_J}L_j}{|n-l-L_j+1|\tau_J}
\]
for any $J\in\mathbb{N}$, $l\in\mathcal{L}_J$, $j\in\mathbb{N}$ and $t>0$, where $C>0$ is a constant which depends only on $\pi_1$ and $\pi_2$.
\end{lemma}

\begin{proof}
By symmetry we may assume $l\in\mathcal{L}_J^+$ without loss of generality.  
Set $\mathcal{X}_{k,l,p,q}(\alpha,\beta)=\chi_{1,k-p}(\alpha)\chi_{2,k+l-q}(\beta)-E\left[\chi_{1,k-p}(\alpha)\chi_{2,k+l-q}(\beta)\right]$. Then, we can rewrite $\widehat{\boldsymbol{\rho}}_{{J-j+1}}(l\tau_J)_{t}-E\left[\widehat{\boldsymbol{\rho}}_{{J-j+1}}(l\tau_J)_{t}|X\right]$ as
\begin{align}
\widehat{\boldsymbol{\rho}}_{{J-j+1}}(l\tau_J)_{t}-E\left[\widehat{\boldsymbol{\rho}}_{{J-j+1}}(l\tau_J)_{t}|X\right]
&=\frac{\tau_J^{-1}}{n-l-L_j+1}\sum_{p,q=0}^{L_j-1}\sum_{\alpha=0}^{\lfloor \tau_J^{-1}t\rfloor-l-p}\sum_{\beta=0}^{\lfloor \tau_J^{-1}t\rfloor-q}h_{j,p}h_{j,q}\nonumber\\
&\qquad\times\sum_{k=(L_j-1)\vee(\alpha+p)\vee(\beta+q-l)}^{\lfloor \tau_J^{-1}t\rfloor-l}\mathcal{X}_{k,l,p,q}(\alpha,\beta)\Delta_{k-p-\alpha}X^1\Delta_{k+l-q-\beta}X^2.\label{lm-expression}
\end{align}
Therefore, the triangular inequality yields
\begin{align*}
&\left\|\widehat{\boldsymbol{\rho}}_{{J-j+1}}(l\tau_J)_{t}-E\left[\widehat{\boldsymbol{\rho}}_{{J-j+1}}(l\tau_J)_{t}|X\right]\right\|_2\\
&\leq\frac{\tau_J^{-1}}{|n-l-L_j+1|}\sum_{p,q=0}^{L_j-1}\sum_{\alpha,\beta=0}^\infty|h_{j,p}h_{j,q}|\left\|\sum_{k=(L_j-1)\vee(\alpha+p)\vee(\beta+q-l)}^{\lfloor \tau_J^{-1}t\rfloor-l}\mathcal{X}_{k,l,p,q}(\alpha,\beta)\Delta_{k-p-\alpha}X^1\Delta_{k+l-q-\beta}X^2\right\|_2.
\end{align*}
Now, by construction $\mathcal{X}_{k,l,p,q}(\alpha,\beta)$ and $\mathcal{X}_{k',l,p,q}(\alpha,\beta)$ are independent if $|k-k'|>\alpha\vee\beta$. Therefore, noting that $X$ and $\mathcal{X}_{k,l,p,q}(\alpha,\beta)$'s are independent, we have
\begin{align*}
&E\left[\left|\sum_{k=(L_j-1)\vee(\alpha+p)\vee(\beta+q-l)}^{\lfloor \tau_J^{-1}t\rfloor-l}\mathcal{X}_{k,l,p,q}(\alpha,\beta)\Delta_{k-p-\alpha}X^1\Delta_{k+l-q-\beta}X^2\right|^2\right]\\
&=\sum_{\begin{subarray}{c}
k,k'=(L_j-1)\vee(\alpha+p)\vee(\beta+q-l)\\
|k-k'|\leq\alpha\vee\beta
\end{subarray}}^{\lfloor \tau_J^{-1}t\rfloor-l}E\left[\mathcal{X}_{k,l,p,q}(\alpha,\beta)\mathcal{X}_{k',l,p,q}(\alpha,\beta)\right]E\left[\Delta_{k-p-\alpha}X^1\Delta_{k+l-q-\beta}X^2\Delta_{k'-p-\alpha}X^1\Delta_{k'+l-q-\beta}X^2\right].
\end{align*}
We have 
\begin{align*}
E\left[\mathcal{X}_{k,l,p,q}(\alpha,\beta)^2\right]
&\leq E\left[\chi_{1,k-p}(\alpha)^2\chi_{2,k+l-q}(\beta)^2\right]
=E\left[\chi_{1,k-p}(\alpha)\right]E\left[\chi_{2,k+l-q}(\beta)\right]\\
&=(1-\pi_1)(1-\pi_2)\pi_1^\alpha\pi_2^\beta,
\end{align*}
and the Burkholder-Davis-Gundy inequality yields
\begin{align*}
\max_{\nu=1,2}\max_{k\geq0}E\left[|\Delta_kX^\nu|^4\right]\leq c_4\tau_J^2
\end{align*}
for some universal constant $c_4>0$, hence we obtain
\begin{align*}
&E\left[\left|\sum_{k=(L_j-1)\vee(\alpha+p)\vee(\beta+q-l)}^{\lfloor \tau_J^{-1}t\rfloor-l}\mathcal{X}_{k,l,p,q}(\alpha,\beta)\Delta_{k-p-\alpha}X^1\Delta_{k+l-q-\beta}X^2\right|^2\right]\\
&\leq\tau_J^{-1}t\cdot \{2(\alpha\vee\beta)+1\}\cdot (1-\pi_1)(1-\pi_2)\pi_1^\alpha\pi_2^\beta\cdot c_4\tau_J^2
=c_4t\{2(\alpha\vee\beta)+1\}(1-\pi_1)(1-\pi_2)\pi_1^\alpha\pi_2^\beta\tau_J.
\end{align*}
Consequently, we conclude that
\begin{align*}
&\left\|\widehat{\boldsymbol{\rho}}_{{J-j+1}}(l\tau_J)_{t}-E\left[\widehat{\boldsymbol{\rho}}_{{J-j+1}}(l\tau_J)_{t}|X\right]\right\|_2\\
&\leq \frac{\tau_J^{-1}}{|n-l-L_j+1|}\sum_{p,q=0}^{L_j-1}|h_{j,p}h_{j,q}|\sum_{\alpha,\beta=0}^\infty\sqrt{c_4t(1-\pi_1)(1-\pi_2)\{2(\alpha\vee\beta)+1\}\pi_1^\alpha\pi_2^\beta\tau_J}\\
&=\frac{C\sqrt{t\tau_J}}{|n-l-L_j+1|\tau_J}\sum_{p,q=0}^{L_j-1}|h_{j,p}h_{j,q}|,
\end{align*}
where $C>0$ is a constant which depends only on $\pi_1$ and $\pi_2$. 
Now, using the inequality
\begin{equation}\label{l1h}
\sum_{p=0}^{L_j-1}|h_{j,p}|
\leq\sqrt{L_j\sum_{p=0}^{L_j-1}h_{j,p}^2}=\sqrt{L_j},
\end{equation}
we complete the proof.
\end{proof}
}

Now we investigate the asymptotic behavior of $E\left[\widehat{\boldsymbol{\rho}}_{{J-j+1}}(l\tau_J)_{t}|X\right]$. For $\nu=1,2$ and $k\geq L_j-1$, we define the variables $Z^\nu_k$ by
$
Z^\nu_k
=E\left[\mathcal{W}^\nu_{jk}|X\right]
=(1-\pi_\nu)\sum_{p=0}^{L_j-1}\sum_{\alpha=0}^{k-p}h_{j,p}\pi_\nu^\alpha\Delta_{k-p-\alpha}X^\nu.
$
Thanks to the independence between $\delta^1_k$'s and $\delta^2_k$'s, we have
\begin{equation}\label{conditional}
E\left[\widehat{\boldsymbol{\rho}}_{{J-j+1}}(l\tau_J)_{t}|X\right]=\frac{\tau_J^{-1}}{n-l-L_j+1}\sum_{k=L_j-1}^{\lfloor \tau_J^{-1}t\rfloor -l}Z^\nu_kZ^\nu_{k+l}
\end{equation}
for $l\geq0$. To analyze the asymptotic behavior of this quantity, we introduce the ``de-volatilized'' version of $Z^\nu_k$ as follows: 
\begin{align*}
\zeta^\nu_k
=(1-\pi_\nu)\sum_{p=0}^{L_j-1}\sum_{\alpha=0}^{k-p}h_{j,p}\pi_\nu^\alpha\Delta_{k-p-\alpha}B^\nu.
\end{align*}
Since $\zeta^\nu_k$'s are centered Gaussian variables, their distribution is completely determined by their covariance structure. To investigate their covariance structure we introduce the following auxiliary quantity for each $\theta\in\mathbb{R}$:
\begin{align*}
\bar{\rho}_j(\theta)
=\int_{-\infty}^\infty D(\lambda)H_{j,L}(\lambda)\Pi(\lambda)e^{\sqrt{-1}\tau_J^{-1}\theta\lambda}f_{J}(\tau_J^{-1}\lambda)d\lambda.
\end{align*}

\begin{lemma}\label{covariance}
\begin{enumerate}[label={\normalfont(\alph*)}]

\item We have
$
\left|E\left[\zeta^\nu_k\zeta^\nu_{k'}\right]\right|
\leq\tau_J(1-\pi_\nu)^2\sum_{\alpha,\beta=0}^\infty\pi_\nu^{\alpha+\beta}\sum_{p=0}^{L_j-1}|h_{j,p}|1_{\{|k'-k-\beta+\alpha|<L_j\}}
$
for any $\nu=1,2$ and any $k,k'\in\mathbb{Z}_+$.


{
\item We have
\[
\max_{k,k',l\in\mathbb{Z}_+:k,k'\geq N+L_j}\left|E\left[\zeta^1_k\zeta^2_{k'+l}\right]-\tau_J\bar{\rho}_j((k'+l-k)\tau_J)\right|
\leq\tau_JL_j(\pi_1^{N}+\pi_2^{N})
\]
for all $J,N\in\mathbb{N}$. 
}

\end{enumerate}
\end{lemma}

\begin{proof}
First, claim (a) follows from the following identity 
\begin{equation*}
E\left[\zeta^\nu_k\zeta^\nu_{k'}\right]
=\tau_J(1-\pi_\nu)^2\sum_{\alpha=0}^k\sum_{\beta=0}^{k'}\sum_{\begin{subarray}{c}
p=0\\
0\leq k'-\beta-(k-p-\alpha)\leq L_j-1
\end{subarray}}^{(L_j-1)\wedge(k-\alpha)}h_{j,p}h_{j,k'-\beta-(k-p-\alpha)}\pi_\nu^{\alpha+\beta}
\end{equation*}
and the inequality $|h_{j,p}|\leq1$. 

Next, using the identity
\begin{equation}\label{formula}
E\left[\Delta_{k}B^1\Delta_{k+l}B^2\right]=\tau_J\int_{-\infty}^\infty D(\lambda)e^{\sqrt{-1}l}f_{J}(\tau_J^{-1}\lambda)d\lambda,
\end{equation}
{which follows from \eqref{s-csd},} 
we have
\begin{align*}
&E\left[\zeta^1_k\zeta^2_{k'+l}\right]\\
&=(1-\pi_1)(1-\pi_2)\sum_{p,q=0}^{L_j-1}\sum_{\alpha=0}^{k-p}\sum_{\beta=0}^{k'+l-q}h_{j,p}h_{j,q}\pi_1^\alpha\pi_2^\beta E\left[\Delta_{k-p-\alpha}B^1\Delta_{k'+l-q-\beta}B^2\right]\\
&=\tau_J(1-\pi_1)(1-\pi_2)\sum_{p,q=0}^{L_j-1}\sum_{\alpha=0}^{k-p}\sum_{\beta=0}^{k'+l-q}h_{j,p}h_{j,q}\pi_1^\alpha\pi_2^\beta\int_{-\infty}^\infty D(\lambda)e^{\sqrt{-1}(k'+l-q-\beta-k+p+\alpha)\lambda}f_{J}(\tau_J^{-1}\lambda)d\lambda.
\end{align*}
{
Meanwhile, from the definitions of $\Pi(\lambda)$ and $H_{j,L}(\lambda)$ we can easily infer that
\[
\Pi(\lambda)=(1-\pi_1)(1-\pi_2)\sum_{\alpha,\beta=0}^\infty\pi_1^\alpha\pi_2^\beta e^{\sqrt{-1}\lambda(\alpha-\beta)},\qquad
H_{j,L}(\lambda)=\sum_{p,q=0}^{L_j-1}h_{j,p}h_{j,q}e^{\sqrt{-1}\lambda(q-p)},
\]
hence we obtain
\begin{align*}
&\bar{\rho}_j((k'+l-k)\tau_J)\\
&=(1-\pi_1)(1-\pi_2)\sum_{p,q=0}^{L_j-1}\sum_{\alpha,\beta=0}^\infty h_{j,p}h_{j,q}\pi_1^\alpha\pi_2^\beta\int_{-\infty}^\infty D(\lambda)e^{\sqrt{-1}(k'+l-k+\alpha-\beta+q-p)\lambda}f_{J}(\tau_J^{-1}\lambda)d\lambda.
\end{align*}
Therefore, we have
\begin{align*}
&\left|E\left[\zeta^1_k\zeta^2_{k'+l}\right]-\tau_J\bar{\rho}_j((k'+l-k)\tau_J)\right|\\
&\leq\tau_J(1-\pi_1)(1-\pi_2)\sum_{p,q=0}^{L_j-1}
\left(\sum_{\alpha=0}^{k-p}\sum_{\beta=k'+l-q+1}^\infty
+\sum_{\alpha=k-p+1}^\infty\sum_{\beta=0}^\infty\right)
|h_{j,p}h_{j,q}|\pi_1^\alpha\pi_2^\beta\int_{-\infty}^\infty D(\lambda)d\lambda\\
&=\tau_J(1-\pi_1)(1-\pi_2)\sum_{p,q=0}^{L_j-1}|h_{j,p}h_{j,q}|
\frac{(1-\pi_1^{k-p+1})\pi_2^{k'+l-q+1}+\pi_1^{k-p+1}}{(1-\pi_1)(1-\pi_2)}
\int_{-\infty}^\infty D(\lambda)d\lambda\\
&\leq\tau_J\sum_{p,q=0}^{L_j-1}|h_{j,p}h_{j,q}|
(\pi_2^{k'+l-L_j}+\pi_1^{k-L_j})
\int_{-\infty}^\infty D(\lambda)d\lambda\quad(\because\pi_1,\pi_2\leq1)\\
&\leq \tau_JL_j(\pi_2^{k'+l-L_j}+\pi_1^{k-L_j})
\int_{-\infty}^\infty D(\lambda)d\lambda\quad(\because\eqref{l1h}).
\end{align*}
Moreover, using formula (3.741.3) of \cite{GR2007} we deduce
\[
\int_{-\infty}^\infty D(\lambda)d\lambda=\frac{2}{\pi}\int_{-\infty}^\infty \frac{\sin^2(\lambda/2)}{\lambda^2}d\lambda=1.
\]
Hence we complete the proof of claim (b) because $\pi_1,\pi_2\leq1$.
}
\end{proof}

From now on we investigate the asymptotic behavior of $\bar{\rho}_j(\theta)$. We need the following auxiliary result.
\begin{lemma}\label{dH-bound}
It holds that
\[
\sup_{\lambda\in[0,\pi]}\left|\frac{d}{d\lambda}H_{j,L}(\lambda)\right|
\leq(2^j-1)c_L,\text{ where }c_L=\binom{L-2}{L/2-1}\frac{L-1}{2^{L-2}}.
\]
Moreover, $c_L=O(\sqrt{L})$ as $L\to\infty$.
\end{lemma}

\begin{proof}
From the proof of Theorem 3 from \cite{Lai1995} we have\footnote{Note that in \cite{Lai1995} the squared gain function of Daubechies' wavelet filter of length $L$ is defined as $H_L(\lambda)/2$.}
\begin{align*}
\frac{d}{d\lambda}G_L(\lambda)
&=-2\cos^{L-2}\left(\frac{\lambda}{2}\right)\binom{L-2}{L/2-1}(L-1)\sin^{L-2}\left(\frac{\lambda}{2}\right)\frac{\sin\lambda}{2}\\
&=-\binom{L-2}{L/2-1}\frac{L-1}{2^{L-2}}\sin^{L-1}(\lambda),
\end{align*}
where we use the double angle formula for the sine function. Since we have $H_L(\lambda)=G_L(\lambda+\pi)$, the first part of the lemma follows from \eqref{hj-formula} and the Leibniz rule. 

The latter part is a consequence of Stirling's formula.
\end{proof}

We can rewrite $\bar{\rho}_j(\theta)$ as
\begin{align*}
\bar{\rho}_j(\theta)
&=\sum_{i=1}^{J+1}R_{i}\int_{\Lambda_{-i}}D(\lambda)H_{j,L}(\lambda)\Pi(\lambda)e^{\sqrt{-1}\tau_J^{-1}(\theta-\theta_{i})\lambda}d\lambda
=:\sum_{i=1}^{J+1}\bar{\rho}_{j,(i)}(\theta).
\end{align*}

{
\begin{lemma}\label{lemma:H-bound}
It holds that
\[
H_{j,L}(\lambda)\leq 2^j,\qquad
D(\lambda)\leq\frac{1}{2\pi},\qquad
\Pi(\lambda)\leq1
\]
for all $\lambda\in\mathbb{R}$ and $j=1,2,\dots$.
\end{lemma}

\begin{proof}
Since $H_L(\lambda)+G_L(\lambda)=2$ for all $\lambda\in\mathbb{R}$ by Eq.(69d) of \cite{PW2000}, we have $H_L(\lambda)\leq2$ and $G_L(\lambda)\leq2$ for all $\lambda\in\mathbb{R}$. This implies the first inequality. 
The second inequality follows from the identity $D(\lambda)=|2\sin(\lambda/2)/\lambda|^2/(2\pi)$ and the Jordan inequality. The third inequality is evident from the definition of $\Pi(\lambda)$. 
\end{proof}
}

\begin{lemma}\label{lemma:RL}
There is a {universal} constant $A$ such that 
\[
\left|\int_{\pi/2^i}^{\pi/2^{i-1}}D(\lambda)H_{j,L}(\lambda)\Pi(\lambda)e^{\sqrt{-1}a\lambda}d\lambda\right|
\leq
{\frac{2^jc_LA}{|a|}}
\]
for any positive integers $i,j$, any even positive integer $L$ and any non-zero real number $a$.
\end{lemma}

\begin{proof}
Integration by parts yields
\begin{multline*}
\int_{\pi/2^{i}}^{\pi/2^{i-1}}D(\lambda)H_{j,L}(\lambda)\Pi(\lambda)e^{\sqrt{-1}a\lambda}d\lambda
=\frac{1}{\sqrt{-1}a}\left[D(\lambda)H_{j,L}(\lambda)\Pi(\lambda)e^{\sqrt{-1}a\lambda}\right]_{\pi/2^i}^{\pi/2^{i-1}}\\
-\frac{1}{\sqrt{-1}a}\int_{\pi/2^i}^{\pi/2^{i-1}}\frac{d}{d\lambda}\left\{D(\lambda)H_{j,L}(\lambda)\Pi(\lambda)\right\}e^{\sqrt{-1}a\lambda}d\lambda.
\end{multline*}
We can easily see that
$
\sup_{\lambda\in\mathbb{R}}\left|\frac{d}{d\lambda}\left\{D(\lambda)\Pi(\lambda)\right\}\right|<\infty,
$
hence the desired result follows from {Lemmas \ref{dH-bound}--\ref{lemma:H-bound}}.
\end{proof}

\begin{lemma}\label{lemma:zero}
Under the assumptions of Theorem \ref{theorem:main}(a), we have $\bar{\rho}_j(\vartheta_J)\to0$ as $J\to\infty$.
\end{lemma}

\begin{proof}
{
Since we have $|\bar{\rho}_{j,(i)}(\vartheta_J)|\leq 2^j(\tau_{i-1}-\tau_i)$ for every $i$ by Lemma \ref{lemma:H-bound}, 
}
it is enough to prove 
\begin{equation}\label{negligible}
\bar{\rho}_{j,(i)}(\vartheta_J)\to0
\end{equation}
as $J\to\infty$ for any fixed $i$ due to the dominated convergence theorem.
By Theorem 1 from \cite{Lai1995}, we have
\begin{equation}\label{lai}
\lim_{L\to\infty}H_{j,L}(\lambda)
=\left\{
\begin{array}{ll}
2^j  & \text{if }\lambda\in(\frac{\pi}{2^j},\frac{\pi}{2^{j-1}}),\\
0  & \text{if }\lambda\in[0,\frac{\pi}{2^j})\cup(\frac{\pi}{2^{j-1}},\pi].
\end{array}
\right.
\end{equation}
Therefore, \eqref{negligible} holds true if $i\neq j$ due to the bounded convergence theorem. On the other hand, if $i=j$, 
{
Lemma \ref{lemma:RL} implies that there is a universal constant $A>0$ such that
\[
\left|\bar{\rho}_{j,(i)}(\vartheta_J)\right|=\left|R_i\int_{\Lambda_{-i}}D(\lambda)H_{j,L}(\lambda)\Pi(\lambda)e^{\sqrt{-1}\tau_J^{-1}(\vartheta_J-\theta_{j})\lambda}d\lambda\right|
\leq\frac{A\tau_J c_L}{|\vartheta_J-\theta_j|},
\]
}
hence we obtain the desired result by assumption and $c_L=O(\sqrt{L})$. 
\end{proof}

\begin{lemma}\label{lemma:wccf}
Under the assumptions of Theorem \ref{theorem:main}(b), we have 
$
\bar{\rho}_j(\vartheta_J)\to 2^jR_{j}\int_{\Lambda_{-j}}D(\lambda)\Pi(\lambda)e^{\sqrt{-1}b\lambda}d\lambda
$
as $J\to\infty$.
\end{lemma}

\begin{proof}
From the above argument it suffices to prove
\[
R_{j}\int_{\Lambda_{-j}}D(\lambda)H_{j,L}(\lambda)\Pi(\lambda)e^{\sqrt{-1}\tau_J^{-1}(l\tau_J-\theta_{j})\lambda}d\lambda
\to 2^jR_{j}\int_{\Lambda_{-j}}D(\lambda)\Pi(\lambda)e^{\sqrt{-1}b\lambda}d\lambda,
\]
which follows from \eqref{lai} and the bounded convergence theorem.
\end{proof}

\begin{proof}[\upshape{\textbf{Proof of Theorem \ref{theorem:main}}}]
Since $\vartheta_J\in\mathcal{G}_J$, it can be written as $\vartheta_J=l\tau_J$ for some $l=l_J\in\mathcal{L}_J$. For the simplicity of presentation we assume that $l_J\in\mathcal{L}_J^+$ for all $J$. 

Since $l\tau_J\in(-\delta,\delta)$ for every $l=l_J$, {every subsequence of $(l\tau_J)_{J\geq1}$ has a converging subsequence}. Therefore, without loss of generality we may assume that $l\tau_J\to\theta$ as $J\to\infty$ for some $\theta\in[-\delta,\delta]$. Note that $\theta=\theta_j$ for case (b) by assumption. 

Set $\mathfrak{c}=0$ for case (a) and
$
\mathfrak{c}=2^jR_{j}\int_{\Lambda_{-j}}D(\lambda)\Pi(\lambda)e^{\sqrt{-1}b\lambda}d\lambda
$
for case (b). We need to prove $\widehat{\boldsymbol{\rho}}_{{J-j+1}}(l\tau_J)_t\xrightarrow{ucp}\Sigma_t(\theta)\mathfrak{c}$ as $J\to\infty$. By {Lemmas \ref{lemma:ucp}--\ref{lemma:c-tight}} it is enough to show that
$\widehat{\boldsymbol{\rho}}_{{J-j+1}}(l\tau_J)_{t}\to\Sigma_{t}(\theta)\mathfrak{c}$ as $J\to\infty$ for any fixed $t>0$. 
{Let $N=N_J$ be a positive integer depending on $J$ so that $\tau_J^wN\to \mathfrak{a}$ for some $w\in(0,1)$ and $\mathfrak{a}\in(0,\infty)$, }
we decompose $\widehat{\boldsymbol{\rho}}_{{J-j+1}}(l\tau_J)_{t}$ as
\begin{align*}
\widehat{\boldsymbol{\rho}}_{{J-j+1}}(l\tau_J)_{t}
&=\left(\widehat{\boldsymbol{\rho}}_{{J-j+1}}(l\tau_J)_{t}-E\left[\widehat{\boldsymbol{\rho}}_{{J-j+1}}(l\tau_J)_{t}|X\right]\right)\\
&\qquad+\left(E\left[\widehat{\boldsymbol{\rho}}_{{J-j+1}}(l\tau_J)_{t}|X\right]-\frac{\tau_J^{-1}}{n-l-L_j+1}\sum_{k=L_j-1}^{\lfloor t\tau_J^{-1}\rfloor-l}\sigma^1_{(k-L_j-N)_+\tau_J}\sigma^2_{\left\{(k-L_j-N)_++l\right\}\tau_J}\zeta^1_k\zeta^2_{k+l}\right)\\
&\qquad+\frac{\tau_J^{-1}}{n-l-L_j+1}\sum_{k=L_j-1}^{\lfloor t\tau_J^{-1}\rfloor-l}\sigma^1_{(k-L_j-N)_+\tau_J}\sigma^2_{\left\{(k-L_j-N)_++l\right\}\tau_J}\zeta^1_k\zeta^2_{k+l}\\
&=:\mathbf{I}_J+\mathbf{II}_J+\mathbf{III}_J.
\end{align*}

First, since $\|\mathbf{I}_J\|_2=O(\sqrt{\tau_J}L)=o(1)$ as $J\to\infty$ by Lemma \ref{lemma:lm-term}, we obtain $\mathbf{I}_J\to^p0$. 

Next, noting \eqref{conditional}, we decompose $\mathbf{II}_J$ as
\begin{align*}
\mathbf{II}_J
&=\frac{\tau_J^{-1}}{n-l-L_j+1}\sum_{k=L_j-1}^{\lfloor t\tau_J^{-1}\rfloor-l}\left\{\left(Z^1_k-\sigma^1_{(k-L_j-N)_+\tau_J}\zeta^1_k\right)Z^2_{k+l}\right.\\
&\left.\hphantom{\frac{\tau_J^{-1}}{n-l-L_j+1}\sum_{k=L_j-1}^{\lfloor t\tau_J^{-1}\rfloor-l}}+\sigma^1_{(k-L_j-N)_+\tau_J}\left(Z^2_{k+l}-\sigma^2_{\left\{(k-L_j-N)_++l\right\}\tau_J}\zeta^2_{k+l}\right)\right\}\\
&=\mathbf{II}_J'+\mathbf{II}_J''.
\end{align*}
We have
$
\left|\mathbf{II}_J'\right|\lesssim\sum_{k=L_j-1}^{\lfloor t\tau_J^{-1}\rfloor-l}\left|Z^1_k-\sigma^1_{(k-L_j-N)_+\tau_J}\zeta^1_k\right|\left|Z^2_{k+l}\right|
$
Since we have
\begin{align*}
&Z^1_k-\sigma^1_{(k-L_j-N)_+\tau_J}\zeta^1_k\\
&=(1-\pi_1)\left\{\sum_{\alpha=0}^{k\wedge N}\pi_1^\alpha\sum_{p=0}^{(L_j-1)\wedge(k-\alpha)}h_{j,p}\int_{(k-p-\alpha)\tau_J}^{(k-p-\alpha+1)\tau_J}(\sigma^1_s-\sigma^1_{(k-L_j-N)_+\tau_J})dB^1_s\right.\\
&\qquad\left.+\sum_{\alpha=N+1}^{k}\sum_{p=0}^{(L_j-1)\wedge(k-\alpha)}h_{j,p}\pi_1^\alpha\left(\Delta_{k-p-\alpha}X^1-\sigma^1_{(k-L_j-N)_+\tau_J}\Delta_{k-p-\alpha}B^1\right)\right\},
\end{align*}
it holds that
\begin{align*}
\left\|Z^1_k-\sigma^1_{i\tau_m}\xi^1_k\right\|_2
&\lesssim\sqrt{\tau_J}\left\|\sup_{s\in[(k-L_j-N)_+\tau_J,(k+1)\tau_J)}\left|\sigma^1_s-\sigma^1_{(k-L_j-N)_+\tau_J}\right|\right\|_2
+\sqrt{\tau_J}\pi_1^{N}
\end{align*}
by the triangle inequality and \eqref{normalize}. On the other hand, since we can rewrite $Z^2_{k+l}$ as
\[
Z^2_{k+l}=(1-\pi_2)\sum_{\beta=0}^{k+l}\pi_2^\beta\sum_{q=0}^{(L_j-1)\wedge(k+l-\beta)}h_{j,q}\Delta_{k+l-q-\beta}X^2,
\]
we have $\|Z^2_{k+l}\|_2\lesssim\sqrt{\tau_J}$ by the triangle inequality, the boundedness of $\sigma^2$ and \eqref{normalize}. Hence we obtain
\begin{align*}
\left\|\mathbf{II}_J'\right\|_1\lesssim\tau_J\sum_{k=L_j-1}^{\lfloor t\tau_J^{-1}\rfloor-l}\left\|\sup_{s\in[(k-L_j-N)_+\tau_J,(k+1)\tau_J)}\left|\sigma^1_s-\sigma^1_{(k-L_j-N)_+\tau_J}\right|\right\|_2+\pi_1^{N}
\end{align*}
by the Schwarz inequality. Since $\sigma^1$ is c\`adl\`ag and bounded, the bounded convergence theorem yields $\left\|\mathbf{II}_J'\right\|_1\to0$, hence $\mathbf{II}_J'\to^p0$. By an analogous argument we can prove $\mathbf{II}_J''\to^p0$, hence we obtain $\mathbf{II}_J\to^p0$.

Finally we prove $\mathbf{III}_J\to^p\Sigma_{t}(\theta)\mathfrak{c}$. It suffices to prove
\begin{equation}\label{p-ut}
\sup_{0\leq s\leq t}\left|\sum_{k=L_j-1}^{\lfloor s\tau_J^{-1}\rfloor-l}\sigma^1_{(k-L_j-N)_+\tau_J}\sigma^2_{\left\{(k-L_j-N)_++l\right\}\tau_J}\zeta^1_k\zeta^2_{k+l}-\Sigma_{s}(\theta)\mathfrak{c}\right|\to^p0.
\end{equation}
Define the process $A^J=(A^J_s)_{s\geq0}$ by 
$
A^J_s=\sum_{k=L_j-1}^{\lfloor s\tau_J^{-1}\rfloor-l}\left(\zeta^1_k\zeta^2_{k+l}-E\left[\zeta^1_k\zeta^2_{k+l}\right]\right).
$
$A^J$ is obviously of (locally) bounded variation. Moreover, since it holds that
$
\|\zeta^\nu_k\|_2\leq(1-\pi_\nu)\sum_{\alpha=0}^\infty\pi_\nu^\alpha\sqrt{\sum_{p=0}^{L_j-1}h_{j,p}^2\tau_J}=\sqrt{\tau_J}
$
by the triangular inequality and \eqref{normalize}, we have
$
E\left[\sum_{k=L_j-1}^{\lfloor s\tau_J^{-1}\rfloor-l}\left|\zeta^1_k\zeta^2_{k+l}-E\left[\zeta^1_k\zeta^2_{k+l}\right]\right|\right]=O(1)
$
for any $s>0$ by the Schwarz inequality. Hence $A^J$ is P-UT by 6.6 from chapter VI of \cite{JS2003}. Moreover, the process 
\[
\left(\sigma^1_{(\lfloor s\tau_J^{-1}\rfloor-L_j-N)_+\tau_J}\sigma^2_{\left\{(\lfloor s\tau_J^{-1}\rfloor-L_j-N)_++l\right\}\tau_J}\right)_{s\geq0}
\] 
converges in probability to the process $(\sigma^1_s\sigma^2_{s+\theta})_{s\geq0}$ for the Skorokhod topology by Proposition 6.37 from chapter VI of \cite{JS2003}. Therefore, according to Theorem 6.22 from chapter VI of \cite{JS2003}, \eqref{p-ut} follows once we show that
\[
\sup_{0\leq s\leq t}\left|\sum_{k=L_j-1}^{\lfloor s\tau_J^{-1}\rfloor-l}\zeta^1_k\zeta^2_{k+l}-\mathfrak{c}(s-\theta)_+\right|\to^p0.
\]
By Lemma \ref{lemma:c-tight} it is enough to prove
$
\sum_{k=L_j-1}^{\lfloor s\tau_J^{-1}\rfloor-l}\zeta^1_k\zeta^2_{k+l}\to^p\mathfrak{c}(s-\theta)_+
$
for any fixed $s$. Moreover, by Lemmas \ref{covariance} and \ref{lemma:zero}--\ref{lemma:wccf} this follows from
\begin{equation}\label{GQF}
\sum_{k=L_j-1}^{\lfloor s\tau_J^{-1}\rfloor-l}\left(\zeta^1_k\zeta^2_{k+l}-E\left[\zeta^1_k\zeta^2_{k+l}\right]\right)\to^p0.
\end{equation}
To prove \eqref{GQF}, let us define the random vector $\boldsymbol{\zeta}$ by
\begin{equation}\label{z}
\boldsymbol{\zeta}=\left(\zeta^1_{L_j-1},\cdots,\zeta^1_{\lfloor\tau_J^{-1}s\rfloor-l},\zeta^2_{L_j-1+l},\cdots,\zeta^2_{\lfloor\tau_J^{-1}s\rfloor}\right)^\top.
\end{equation}
Then we have
$
\sum_{k=L_j-1}^{\lfloor s\tau_J^{-1}\rfloor-l}\left(\zeta^1_k\zeta^2_{k+l}-E\left[\zeta^1_k\zeta^2_{k+l}\right]\right)=\boldsymbol{\zeta}^\top A_J\boldsymbol{\zeta}-E\left[\boldsymbol{\zeta}^\top A_J\boldsymbol{\zeta}\right],
$
where
\[
A_J=\left(
\begin{array}{cc}
0  &  \mathsf{{E}}_{\lfloor\tau_J^{-1}s\rfloor-l-L_j+2}    \\
\mathsf{{E}}_{\lfloor\tau_J^{-1}s\rfloor-l-L_j+2}  &   0
\end{array}
\right).
\]
Therefore, the proof of \eqref{GQF} is completed once we show that $\variance[\boldsymbol{\zeta}^\top A_J\boldsymbol{\zeta}]\to0$ as $J\to\infty$. Since $\boldsymbol{\zeta}$ is centered Gaussian, we have $\variance[\boldsymbol{\zeta}^\top A_J\boldsymbol{\zeta}]=2\trace[(\Sigma_JA_J)^2]$ from {Eq.(4.4) of \cite{Davies1973}}, 
where $\Sigma_J$ denotes the covariance matrix of $\boldsymbol{\zeta}$. Since $\trace[(\Sigma_JA_J)^2]\leq\|\Sigma_JA_J\|_F^2\leq\|\Sigma_J\|_F^2$ by Appendix II(ii)--(iii) from \cite{Davies1973}, it is enough to prove $\|\Sigma_J\|_F^2\to0$. 

First, by Lemma \ref{covariance}(a) and \eqref{l1h} we have
$
\sum_{k,k'=L_j-1}^{\lfloor s\tau_J^{-1}\rfloor-l}\left(\left|E\left[\zeta^1_k\zeta^1_{k'}\right]\right|^2+\left|E\left[\zeta^1_{k+l}\zeta^1_{k'+l}\right]\right|^2\right)
=O\left(\tau_JL_j^2\right)=o(1).
$
Next, by Lemma \ref{covariance}(b) we have
$
\sum_{k,k'=L_j-1}^{\lfloor s\tau_J^{-1}\rfloor-l}\left|E\left[\zeta^1_k\zeta^2_{k'+l}\right]\right|^2=\tau_J^2\sum_{k,k'=L_j+N}^{\lfloor s\tau_J^{-1}\rfloor-l}\bar{\rho}_j((k'+l-k)\tau_J)^2+o(1).
$
Therefore, to complete the proof of $\|\Sigma_J\|_F^2\to0$, it is enough to show that
\begin{equation}\label{rho2vanish}
\tau_J^2\sum_{k,k'=L_j+N}^{\lfloor s\tau_J^{-1}\rfloor-l}\bar{\rho}_j((k'+l-k)\tau_J)^2\to0.
\end{equation}
We rewrite the target quantity as
\begin{align*}
&\tau_J^2\sum_{k,k'=L_j+N}^{\lfloor s\tau_J^{-1}\rfloor-l}\bar{\rho}_j((k'+l-k)\tau_J)^2\\
&=\tau_J^2\sum_{i_1,i_2=1}^{J+1}\sum_{k,k'=L_j+N}^{\lfloor s\tau_J^{-1}\rfloor-l}\prod_{r=1}^2R_{i_r}\int_{\Lambda_{-i}}D(\lambda)H_{j,L}\left(\lambda\right)\Pi(\lambda)e^{\sqrt{-1}\lambda(k'+l-k-\theta_{i_r}\tau_J^{-1})}d\lambda\\
&=:\sum_{i_1,i_2=1}^{J+1}\Xi_J(i_1,i_2).
\end{align*}
Since we have $\left|\Xi_J(i_1,i_2)\right|\lesssim(\tau_{i_1-1}-\tau_{i_1})(\tau_{i_2-1}-\tau_{i_2})$, we obtain \eqref{rho2vanish} by the dominated convergence theorem once we prove $\Xi_J(i_1,i_2)\to0$ for any fixed $i_1,i_2$. By Lemma \ref{lemma:RL} we have
\begin{align*}
\left|\Xi_J(i_1,i_2)\right|
&\lesssim\tau_J+\tau_J^2\sum_{\begin{subarray}{c}
k,k'=N+L_j\\
\theta_{i_1},\theta_{i_2}\neq (k'+l-k)\tau_J
\end{subarray}}^{\lfloor \tau_J^{-1}t\rfloor-l}\prod_{r=1}^2\frac{c_L}{\left|k'+l-k-\theta_{i_r}\tau_J^{-1}\right|}\\
&\lesssim\tau_J+\tau_J^2c_L^2\sum_{r=1}^2\sum_{\begin{subarray}{c}
k,k'=N+L_j\\
\theta_{i_r}\neq (k'+l-k)\tau_J
\end{subarray}}^{\lfloor \tau_J^{-1}t\rfloor-l}\frac{1}{\left|k'+l-k-\theta_{i_r}\tau_J^{-1}\right|^2}
\lesssim\tau_J+\tau_Jc_L^2.
\end{align*}
Since $c_L=O(\sqrt{L})$, we conclude that $\Xi_J(i_1,i_2)\to^p0$.

Consequently, we complete the proof of the theorem.
\end{proof}

\subsection{Proof of Theorem \ref{HRY}}\label{proof:HRY}

Similarly to Section \ref{proof:theorem1}, a localization procedure allows us to assume that both $\sigma^1$ and $\sigma^2$ are bounded as well as there is a constant $K>0$ such that 
\begin{equation}\label{holder}
|\sigma^1_t-\sigma^1_s|+|\sigma^2_t-\sigma^2_s|\leq K|t-s|^\gamma
\end{equation}
for any $t,s\geq0$. 

{
For the proof it is convenient to introduce the Orlicz norms based on the functions $\psi_p(x)=e^{x^p}-1$, $x\in\mathbb{R}$, for $p\geq1$: For a random variable $Y$, we define
\[
\|Y\|_{\psi_p}:=\inf\left\{C>0:E\left[\psi_p\left(\frac{|Y|}{C}\right)\right]\leq1\right\}.
\]
In this paper only the cases $p=1$ and $p=2$ are necessary. We refer to Section 2.2 of \cite{VW1996} for an exposition of Orlicz norms.  

{
Before starting the main body of the proof, we give an overview of our strategy. 
As is suggested by the similarity of Theorem \ref{HRY} and \cite[Theorem 1]{HRY2013}, we adopt a strategy analogous to the proof of the latter result. 
Since \cite[Propositions 3--4]{HRY2013} play key roles in its proof, we need to establish their analogs in our context. 
In fact, an analogous result to \cite[Proposition 3]{HRY2013} has already been obtained in Theorem \ref{theorem:main}(b), so we only need to prove an analog of \cite[Proposition 4]{HRY2013}. 
More precisely, we will show the following result:
\begin{proposition}\label{lemma:ld}
Under the assumptions of Theorem \ref{HRY}, we have
$
\max_{\theta\in\mathcal{G}_{J}:v_{J}^{-1}|\theta-\theta_j|> \varepsilon}\left|\widehat{\rho}_{{J-j+1}}(\theta)\right|\to^p0
$
as $J\to\infty$ for any $\varepsilon>0$.
\end{proposition}
The strategy of the proof of Proposition \ref{lemma:ld} is as follows. By symmetry and the triangular inequality, it suffices to prove the following equations:
\begin{align}
&\max_{l\in\mathcal{L}^+_{J}:v_{J}^{-1}|l\tau_J-\theta_j|> \varepsilon}\left|\widehat{\rho}_{{J-j+1}}(l\tau_J)-E\left[\widehat{\rho}_{{J-j+1}}(l\tau_J)|X\right]\right|\to^p0,\label{LD-1}\\
&\max_{l\in\mathcal{L}^+_{J}:v_{J}^{-1}|l\tau_J-\theta_j|> \varepsilon}\left|E\left[\widehat{\rho}_{{J-j+1}}(l\tau_J)|X\right]\right|\to^p0.\label{LD-2}
\end{align}
We prove \eqref{LD-1} in Section \ref{sec:proofLD-1}, where the proof is achieved by establishing higher moment bounds for the interpolation errors $\widehat{\rho}_{{J-j+1}}(l\tau_J)-E\left[\widehat{\rho}_{{J-j+1}}(l\tau_J)|X\right]$ in Lemma \ref{lemma:lm-term2}. 

Meanwhile, we prove \eqref{LD-2} in Section \ref{sec:proofLD-2}. 
The proof of \eqref{LD-2} is more involved but described as follows. 
First, we divide the interval $[0,n\tau_J]$ into sub-blocks of the length $\tau_m$ (plus negligible sub-blocks appearing at the edge of the interval) and approximate the volatility processes by the simple processes that are obtained by freezing the volatility processes at the beginning of each sub-block. 
More formally, for any positive integers $m\leq J$, $N$ and $i$, we set
\begin{align}
I_{m,N}(i)=I_{m,N}^{(J)}(i)
&=\{k\in\mathbb{Z}_+:(k-L_j-N)\tau_J\in[i\tau_m,(i+1)\tau_m)\}\nonumber\\
&=\{2^{J-m}i+L_j+N,2^{J-m}i+L_j+N+1,\dots,2^{J-m}(i+1)+L_j+N-1\}.\label{notation-I}
\end{align} 
Then we approximate $Z_k^\nu$ with $k\in I_{m,N}(i)$ by $\sigma^\nu_{i\tau_m}\zeta^\nu_k$. 
As a result, $E\left[\widehat{\rho}_{{J-j+1}}(l\tau_J)|X\right]$ is approximated by a sum of random variables of the form $\sum_{k\in I_{m,N}(i)}\zeta_k^\nu\zeta_{k+l}^\nu$ multiplied by a bounded random variable. 
Here, we need to carefully choose the integers $m$ and $N$ so that the above approximation is good enough while the exponential moment bounds for the approximators derived in the next step are sufficiently sharp. 
We remark that this type of procedure (the so-called blocking technique) is standard in the literature on limit theorems for power variations (see e.g.~the proof of Theorems 3--4 from \cite{BNCP2011}). 

In the next step, we prove the variables $\sum_{k\in I_{m,N}(i)}\zeta_k^\nu\zeta_{k+l}^\nu$ are negligible. For this purpose we decompose them into $\sum_{k\in I_{m,N}(i)}(\zeta_k^\nu\zeta_{k+l}^\nu-E[\zeta_k^\nu\zeta_{k+l}^\nu])$'s and $\sum_{k\in I_{m,N}(i)}E[\zeta_k^\nu\zeta_{k+l}^\nu]$'s. 
The negligibility of the latter ones immediately follows from the estimates established in the proof of Theorem \ref{theorem:main}. Specifically, Lemmas \ref{covariance}(b) and \ref{lemma:zero} can be used.   
In the meantime, we show that the former ones are negligible by establishing their exponential moment bounds in Lemma \ref{gqf}. 
}

{\subsubsection{Proof of (\ref{LD-1})}\label{sec:proofLD-1}}

{As is remarked in the above, we prove \eqref{LD-1} by establishing higher moment bounds for the variables $\widehat{\rho}_{{J-j+1}}(l\tau_J)-E\left[\widehat{\rho}_{{J-j+1}}(l\tau_J)|X\right]$ (cf.~Lemma \ref{lemma:lm-term}).}
For this purpose we will use the following moment inequality for $m$-dependent variables obtained as a consequence of Theorem 2 of \cite{DL1999}. Recall that a sequence $(Y_k)_{k\in\mathbb{N}}$ of random variables is said to be \textit{$m$-dependent} for some integer $m\geq0$ if $\{Y_k:k\leq i\}$ and $\{Y_k:k\geq j\}$ are independent for any positive integers $i,j$ such that $j-i>m$. 
\begin{lemma}\label{lemma:DL}
Let $(Y_k)_{k\in\mathbb{N}}$ be a sequence of centered random variables which is $m$-dependent for some $m\in\mathbb{N}$. Suppose that there is a constant $A>0$ and a positive even integer $q$ such that $E[(Y_k)^{q}]\leq A$ for every $k$. Then we have
\begin{align*}
E\left[\left(\sum_{k=1}^NY_k\right)^{q}\right]
\leq\frac{(2q-2)!}{(q-1)!}A\{2N(m+1)\}^{q/2}
\end{align*}
for any integer $N$ greater than $m$.
\end{lemma}

\begin{proof}
Take positive integers $p,s$ satisfying $s<p\leq q$ and positive integers $k_1\leq\cdots\leq k_p$ arbitrarily. If $k_{s+1}-k_s> m$, we have 
\[
\covariance[Y_{k_1}\cdots Y_{k_s},Y_{k_{s+1}}\cdots Y_{k_p}]=0
\]
by the $m$-dependence of $(Y_k)_{k\in\mathbb{N}}$. Otherwise, we have
\begin{align*}
\left|\covariance[Y_{k_1}\cdots Y_{k_s},Y_{k_{s+1}}\cdots Y_{k_p}]\right|
&\leq\left|E[Y_{k_1}\cdots Y_{k_p}]\right|
+\left|E[Y_{k_1}\cdots Y_{k_s}]E[Y_{k_{s+1}}\cdots Y_{k_p}]\right|\\
&\leq 2\|Y_{k_1}\|_{p}\cdots \|Y_{k_p}\|_{p}
\leq 2A^{\frac{p}{q}}
\end{align*}
by the generalized H\"older and Lyapunov inequalities. Therefore, the dependence coefficients $C_{r,p}$ associated to $(Y_k)$ (see Definition 2 of \cite{DL1999}) satisfy
\[
C_{r,p}\leq 2A^{\frac{p}{q}}1_{\{r\leq m\}}.
\]
Therefore, applying Theorem 2 of \cite{DL1999} with $C=2$, $\gamma=\log(A^{\frac{1}{q}})$, $M=1$ and $\theta_r=1_{\{r\leq m\}}$, we obtain
\begin{align*}
E\left[\left(\sum_{k=1}^NY_k\right)^{q}\right]
&\leq\frac{(2q-2)!}{(q-1)!}A\left\{(2N(m+1))^{q/2}\vee(2N(m+1)^{q-1})\right\}\\
&\leq\frac{(2q-2)!}{(q-1)!}A(2N)^{q/2}(m+1)^{q/2}
\end{align*} 
for any integer $N>m$. Hence we complete the proof. 
\end{proof}

We also need the following moment inequality for the maximum of increments of Brownian martingales:
\begin{lemma}\label{BM-max}
There is a universal constant $c>0$ such that
\[
\max_{\nu=1,2}\left\|\max_{0\leq k\leq N-1}\left|\Delta_kX^\nu\right|\right\|_p
\leq cp!A\sqrt{\tau_J\log(N+1)}
\]
for any $p\geq1$, $J,N\in\mathbb{N}$ and any constant $A>0$ such that $\max_{\nu=1,2}\sup_{t\geq0}|\sigma_t^\nu|\leq A$ a.s.
\end{lemma}

\begin{proof}
By Proposition 4.2.3 of \cite{BJY1986} we have
\[
P(M_{{\mathsf{T}}}\geq x,\langle M\rangle_{{\mathsf{T}}}\leq y^2)\leq\exp\left(-\frac{x^2}{2y^2}\right)
\]
for any continuous martingale $M$, stopping time {$\mathsf{T}$} and $x,y>0$. Therefore, we have
\[
P(|\Delta_kX^\nu|\geq x)\leq 2\exp\left(-\frac{x^2}{2A^2\tau_J}\right)
\]
for all $J\in\mathbb{N}$, $k=0,1\dots$ and $x>0$. 
Thus, by Lemma 2.2.1 of \cite{VW1996} we have
\[
\|\Delta_kX^\nu\|_{\psi_2}\leq\sqrt{6A^2\tau_J}
\]
for all $k=0,1\dots$ and $J\in\mathbb{N}$. Hence Lemma 2.2.2 of \cite{VW1996} implies that there is a universal constant $c_0>0$ such that
\[
\left\|\max_{0\leq k\leq N-1}|\Delta_kX^\nu|\right\|_{\psi_2}\leq c_0\sqrt{6A^2\tau_J\log(N+1)}
\]
for all $k=0,1\dots$ and $J,N\in\mathbb{N}$. Now the desired result follows from the inequality $\|Y\|_p\leq p!\sqrt{\log 2}\|Y\|_{\psi_2}$ which holds for any random variable $Y$ and $p\geq1$ (cf.~page 95 of \cite{VW1996}). 
\end{proof}

Now we are ready to prove the following result:
\begin{lemma}\label{lemma:lm-term2}
Under the assumptions of Theorem \ref{HRY}, we have 
\[
\left\|\widehat{\boldsymbol{\rho}}_{{J-j+1}}(l\tau_J)_{t}-E\left[\widehat{\boldsymbol{\rho}}_{{J-j+1}}(l\tau_J)_{t}|X\right]\right\|_r
\leq \frac{C_1\sqrt{t\tau_J}L_j\log(\tau_J^{-1}t+1)}{|n-l-L_j+1|\tau_J}
\]
for any $J\in\mathbb{N}$, $l\in\mathcal{L}_J$, $j\in\mathbb{N}$, $t>0$ and any positive even integer $r$, where $C_1>0$ is a constant which depends only on $r,\sigma^1,\sigma^2,\pi_1,\pi_2$. 
\end{lemma}

\begin{proof}
By symmetry we may assume $l\in\mathcal{L}_J^+$ without loss of generality.  

We use the same notation as in the proof of Lemma \ref{lemma:lm-term}. 
Then, \eqref{lm-expression} and the Minkowski inequality yield
\begin{align*}
&\left\|\widehat{\boldsymbol{\rho}}_{{J-j+1}}(l\tau_J)_{t}-E\left[\widehat{\boldsymbol{\rho}}_{{J-j+1}}(l\tau_J)_{t}|X\right]\right\|_r\\
&\leq\frac{\tau_J^{-1}}{|n-l-L_j+1|}\sum_{p,q=0}^{L_j-1}\sum_{\alpha,\beta=0}^\infty|h_{j,p}h_{j,q}|\left\|\sum_{k=(L_j-1)\vee(\alpha+p)\vee(\beta+q-l)}^{\lfloor \tau_J^{-1}t\rfloor-l}\mathcal{X}_{k,l,p,q}(\alpha,\beta)\Delta_{k-p-\alpha}X^1\Delta_{k+l-q-\beta}X^2\right\|_r.
\end{align*}
By construction, the sequence $(\mathcal{X}_{k,l,p,q}(\alpha,\beta)\Delta_{k-p-\alpha}X^1\Delta_{k+l-q-\beta}X^2)_k$ is $(\alpha\vee\beta)$-dependent conditionally on $X$ and
\[
E\left[|\Delta_{k-p-\alpha}X^1\Delta_{k+l-q-\beta}X^2\mathcal{X}_{k,l,p,q}(\alpha,\beta)|^{r}\mid X\right]\leq 2^{r}\pi_1^\alpha\pi_2^\beta\sum_{\nu=1}^2\max_{1\leq m\leq\lfloor\tau_J^{-1}t\rfloor}|\Delta_mX^\nu|^{2r}\quad\text{a.s.}
\] 
Therefore, by Lemma \ref{lemma:DL} we obtain
\begin{align*}
&E\left[\left|\sum_{k=(L_j-1)\vee(\alpha+p)\vee(\beta+q-l)}^{\lfloor \tau_J^{-1}t\rfloor-l}\mathcal{X}_{k,l,p,q}(\alpha,\beta)\Delta_{k-p-\alpha}X^1\Delta_{k+l-q-\beta}X^2\right|^r\mid X\right]\\
&\leq\frac{(2r-2)!}{(r-1)!}2^r\pi_1^\alpha\pi_2^\beta\{2\tau_J^{-1}t(\alpha\vee\beta+1)\}^{r/2}\sum_{\nu=1}^2\max_{0\leq m\leq\lfloor\tau_J^{-1}t\rfloor-1}|\Delta_mX^\nu|^{2r}
\quad\text{a.s.}
\end{align*}
Hence, by Lemma \ref{BM-max} there is a constant $C>0$ which depends only on $r,\sigma^1,\sigma^2$ such that
\begin{align*}
&E\left[\left|\sum_{k=(L_j-1)\vee(\alpha+p)\vee(\beta+q-l)}^{\lfloor \tau_J^{-1}t\rfloor-l}\mathcal{X}_{k,l,p,q}(\alpha,\beta)\Delta_{k-p-\alpha}X^1\Delta_{k+l-q-\beta}X^2\right|^r\right]\\
&\leq C\pi_1^\alpha\pi_2^\beta\{\tau_J^{-1}t(\alpha\vee\beta+1)\}^{r/2}\tau_J^r\log^r(\tau_J^{-1}t+1).
\end{align*}
Consequently, we obtain 
\begin{align*}
\left\|\widehat{\boldsymbol{\rho}}_{{J-j+1}}(l\tau_J)_{t}-E\left[\widehat{\boldsymbol{\rho}}_{{J-j+1}}(l\tau_J)_{t}|X\right]\right\|_r
&\leq\frac{C\sqrt{t\tau_J}\log(\tau_J^{-1}t+1)}{|n-l-L_j+1|\tau_J}\sum_{p,q=0}^{L_j-1}|h_{j,p}h_{j,q}|\sum_{\alpha,\beta=0}^\infty\pi_1^{\alpha/r}\pi_2^{\beta/r}\sqrt{\alpha\vee\beta+1}.
\end{align*}
Therefore, there is a constant $C_1>0$ which depends only on $r,\sigma^1,\sigma^2,\pi_1,\pi_2$ such that
\begin{align*}
\left\|\widehat{\boldsymbol{\rho}}_{{J-j+1}}(l\tau_J)_{t}-E\left[\widehat{\boldsymbol{\rho}}_{{J-j+1}}(l\tau_J)_{t}|X\right]\right\|_r
&\leq\frac{C_1\sqrt{t\tau_J}\log(\tau_J^{-1}t+1)}{|n-l-L_j+1|\tau_J}\sum_{p,q=0}^{L_j-1}|h_{j,p}h_{j,q}|.
\end{align*}
Now the desired result follows from \eqref{l1h}.
\end{proof}
}

\begin{proof}[Proof of (\ref{LD-1})]
Take $\varepsilon>0$ arbitrarily. Then, for any even integer $r$ the Markov inequality yields
\begin{align*}
&P\left(\max_{l\in\mathcal{L}^+_{J}:v_{J}^{-1}|l\tau_J-\theta_j|> \varepsilon}\left|\widehat{\rho}_{{J-j+1}}(l\tau_J)-E\left[\widehat{\rho}_{{J-j+1}}(l\tau_J)|X\right]\right|>\varepsilon\right)\\
&\leq\varepsilon^{-r}E\left[\max_{l\in\mathcal{L}^+_{J}:v_{J}^{-1}|l\tau_J-\theta_j|> \varepsilon}\left|\widehat{\rho}_{{J-j+1}}(l\tau_J)-E\left[\widehat{\rho}_{{J-j+1}}(l\tau_J)|X\right]\right|^r\right]\\
&\leq\varepsilon^{-r}\sum_{l\in\mathcal{L}^+_{J}}E\left[\left|\widehat{\rho}_{{J-j+1}}(l\tau_J)-E\left[\widehat{\rho}_{{J-j+1}}(l\tau_J)|X\right]\right|^r\right]
\leq\varepsilon^{-r}\#\mathcal{L}^+_{J}\max_{l\in\mathcal{L}^+_{J}}E\left[\left|\widehat{\rho}_{{J-j+1}}(l\tau_J)-E\left[\widehat{\rho}_{{J-j+1}}(l\tau_J)|X\right]\right|^r\right].
\end{align*}
We can take sufficiently large $r$ such that $2/r<1-\kappa$. Then we have 
\begin{align*}
\#\mathcal{L}^+_{J}\max_{l\in\mathcal{L}^+_{J}}E\left[\left|\widehat{\rho}_{{J-j+1}}(l\tau_J)-E\left[\widehat{\rho}_{{J-j+1}}(l\tau_J)|X\right]\right|^r\right]
=O\left(\left(\tau_J^{1-2/r}L^2{\log^2\tau_J}\right)^{r/2}\right)=o(1)
\end{align*}
by {Lemma \ref{lemma:lm-term2}}. This yields the desired result. 
\end{proof}

{\subsubsection{Proof of (\ref{LD-2})}\label{sec:proofLD-2}}

{
{
As is stated at the beginning of this subsection, we need to establish an exponential moment bound for the variable of the form 
\[
Q_{m,N,l}(i):=\sum_{k\in I_{m,N}(i)}\left(\zeta^1_k\zeta^2_{k+l}-E\left[\zeta^1_k\zeta^2_{k+l}\right]\right),
\]
where $m\leq J$, $N$, $i$ and $l\in\mathcal{L}_J^+$ are given positive integers and $I_{m,N}(i)$ is defined by \eqref{notation-I}. 
}

We begin by remarking the following result: 
\begin{lemma}\label{hyper}
There is a universal constant $C>0$ such that $\|Y\|_{\psi_1}\leq C\|Y\|_2$ for any quadratic polynomial $Y$ of Gaussian variables. 
\end{lemma}

\begin{proof}
Since $Y=0$ a.s.~if $\|Y\|_2=0$, it suffices to consider the case $\|Y\|_2\neq0$. Hence, we may assume $\|Y\|_2=1$ without loss of generality. By Theorem 6.7 of \cite{Janson1997}, there is a universal constant $c>0$ such that $P(|Y|>x)\leq e^{-cx}$ for any $x\geq2$. Setting $K=1\vee e^{2c}$, we have $Ke^{-cx}\geq1$ for all $x\leq2$. Hence we have $P(|Y|>x)\leq Ke^{-cx}$ for any $x>0$ and thus we have $\|Y\|_{\psi_1}\leq (1+K)/c$ by Lemma 2.2.1 of \cite{VW1996}. Therefore, the desired result holds true with $C=(1+K)/c$. 
\end{proof}

\begin{lemma}\label{gqf}
Under the assumptions of Theorem \ref{HRY}, there is a universal constant $C>0$ such that
\[
\left\|{Q_{m,N,l}(i)}\right\|_{\psi_1}
\leq C\sqrt{\tau_m\tau_J}\left(L_j^{3/4}
+\sqrt{L_j(\pi_1^N+\pi_2^N)}
+\sqrt{L_j}|\log\tau_J|\right)
\]
for any positive integers $J$, $m\leq {J}$, $N$, $i$ and {$l\in\mathcal{L}_J^+$}.
\end{lemma}

\begin{proof}
By Lemma \ref{hyper} it suffices to prove
\[
\|Q_{m,N,l}(i)\|_2\leq C'\sqrt{\tau_m\tau_J}\left(L_j^{3/4}
+\sqrt{L_j(\pi_1^N+\pi_2^N)}
+\sqrt{L_j}|\log\tau_J|\right)
\]
where $C'>0$ is a universal constant. 
Let $\Sigma_{m,N,l}(i)$ be the covariance matrix of
\[
((\zeta^1_k)_{k\in I_{m,N}(i)},(\zeta^2_{k+l})_{k\in I_{m,N}(i)})^\top
\] 
for every $i$, and set 
$
C_{m,N,l}(i)=\Sigma_{m,N,l}(i)^{1/2}A_m\Sigma_{m,N,l}(i)^{1/2},
$ 
where
\[
A_m=\left(
\begin{array}{cc}
0  &  \mathsf{{E}}_{2^{J-m}}    \\
\mathsf{{E}}_{2^{J-m}}  &   0
\end{array}
\right).
\]
From Eq.(4.4) of \cite{Davies1973} we have $\|Q_{m,N,l}(i)\|_2=\sqrt{2}\|C_{m,N,l}(i)\|_F$, hence it is enough to prove
\begin{equation}\label{matrix}
\|C_{m,N,l}(i)\|_F\leq C''\sqrt{\tau_m\tau_J}\left(L_j^{3/4}
+\sqrt{L_j(\pi_1^N+\pi_2^N)}
+\sqrt{L_j}|\log\tau_J|\right),
\end{equation}
where $C''>0$ is a universal constant. 
By Appendix II(ii)--(iii) from \cite{Davies1973}, we have
\begin{align*}
\|C_{m,N,l}(i)\|_F^2&\leq\|\Sigma_{m,N,l}(i)\|_F^2\\
&=\sum_{k,k'\in I_{m,N}(i)}E\left[\zeta^1_k\zeta^1_{k'}\right]^2+\sum_{k,k'\in I_{m,N}(i)}E\left[\zeta^2_{k+l}\zeta^2_{k'+l}\right]^2+2\sum_{k,k'\in I_{m,N}(i)}E\left[\zeta^1_{k}\zeta^2_{k'+l}\right]^2.
\end{align*}
Since we have for every $k$
\[
\|\zeta^\nu_k\|_2\leq(1-\pi_1)\sum_{\alpha=0}^{k-p}\pi_\nu^\alpha\left\|\sum_{p=0}^{L_j-1}h_{j,p}\Delta_{k-p-\alpha}B^\nu\right\|\leq\sqrt{\tau_J},
\]
where we use \eqref{normalize}, we obtain $|E\left[\zeta^\nu_{k}\zeta^\nu_{k'}\right]|\leq\tau_J$ for any $k,k'$ by the Schwarz inequality. 
Combining this estimate with Lemma \ref{covariance}(a) and \eqref{l1h}, we obtain
\begin{align*}
&\sum_{k,k'\in I_{m,N}(i)}E\left[\zeta^1_k\zeta^1_{k'}\right]^2+\sum_{k,k'\in I_{m,N}(i)}E\left[\zeta^2_{k+l}\zeta^2_{k'+l}\right]^2\\
&\leq\tau_J^2(2L_j-1)\# I_{m,N}(i)\sum_{\nu=1,2}(1-\pi_\nu)^2\sum_{\alpha,\beta=0}^\infty\pi_\nu^{\alpha+\beta}\sum_{p=0}^{L_j-1}|h_{j,p}|
\leq 4\tau_J\tau_mL_j^{3/2}.
\end{align*}
Meanwhile, Lemma \ref{lemma:H-bound}, the Parseval identity and \eqref{normalize} yield
\[
|\bar{\rho}_j(\theta)|\leq\frac{1}{2\pi}\int_{-\pi}^\pi H_{j,L}(\lambda)d\lambda
=1.
\]
Moreover, $|E\left[\zeta^1_{k}\zeta^2_{k'}\right]|\leq\tau_J$ for any $k,k'$ by the Schwarz inequality. Combining these estimates with Lemma \ref{covariance}(b) we obtain
\begin{align*}
\left|\sum_{k,k'\in I_{m,N}(i)}E\left[\zeta^1_{k}\zeta^2_{k'+l}\right]^2
-\tau_J^2\sum_{k,k'\in I_{m,N}(i)}\bar{\rho}_j((k'+l-k)\tau_J)^2\right|\leq\tau_m\tau_JL_j(\pi_1^N+\pi_2^N).
\end{align*}
Since it holds that
\begin{align*}
&\sum_{k,k'\in I_{m,N}(i)}\bar{\rho}_j((k'+l-k)\tau_J)^2
=\sum_{k,k'\in I_{m,N}(0)}\bar{\rho}_j((k'+l-k)\tau_J)^2,
\end{align*}
the proof of \eqref{matrix} is completed once we show that
\begin{equation}\label{frobenius}
\tau_J^2\sum_{k,k'\in I_{m,N}(0)}\bar{\rho}_j((k'+l-k)\tau_J)^2\leq A\tau_J\tau_m(\log\tau_J)^2L_j
\end{equation}
for some universal constant $A>0$. 
We have
\begin{align*}
&\tau_J^2\sum_{k,k'\in I_{m,N}(0)}\bar{\rho}_j((k'+l-k)\tau_J)^2\\
&\leq\tau_J^2\sum_{i_1,i_2=1}^{J+1}\sum_{k,k'\in I_{m,N}(0)}\prod_{r=1}^2\left|\int_{\Lambda_{-i_r}}D(\lambda)H_{j,L}(\lambda)\Pi(\lambda)e^{\sqrt{-1}\tau_J^{-1}((k'+l-k)\tau_J-\theta_{i_r})\lambda}d\lambda\right|
\\
&\leq 2\tau_J\tau_m+\tau_J^2\max_{l\in\mathcal{L}_J^+}\sum_{i_1,i_2=1}^{J+1}\sum_{\begin{subarray}{c}
k,k'\in I_{m,N}(0)\\
\theta_{i_1},\theta_{i_2}\neq (k'+l-k)\tau_J
\end{subarray}}\prod_{r=1}^2\left|\int_{\Lambda_{-i_r}}D(\lambda)H_{j,L}(\lambda)\Pi(\lambda)e^{\sqrt{-1}\tau_J^{-1}((k'+l-k)\tau_J-\theta_{i_r})\lambda}d\lambda\right|,
\end{align*}
where we use Lemma \ref{lemma:H-bound}, the Parseval identity and \eqref{normalize} to obtain the last inequality. Therefore, Lemma \ref{lemma:RL} yields
\begin{align*}
&\tau_J^2\sum_{k,k'\in I_{m,N}(0)}\bar{\rho}_j((k'+l-k)\tau_J)^2\\
&\leq 2\tau_J\tau_m+2^jC_1\tau_J^2\sum_{i_1,i_2=1}^{J+1}\sum_{\begin{subarray}{c}
k,k'\in I_{m,N}(0)\\
\theta_{i_1},\theta_{i_2}\neq (k'+l-k)\tau_J
\end{subarray}}\prod_{r=1}^2\frac{c_L}{\left|k'+l-k-\theta_{i_r}\tau_J^{-1}\right|}\\
&\leq2\tau_J\tau_m+2^jC_2\tau_J\tau_m(J+1)^2c_L^2
\leq 2^jC_3\tau_J\tau_m(\log\tau_J)^2c_L^2
\end{align*}
for some universal constants $C_1,C_2,C_3>0$. Hence we obtain \eqref{frobenius} because $2^jc_L\leq A'L_j$ for some universal constant $A'>0$, and thus we complete the proof.
\end{proof}
}

\begin{proof}[Proof of (\ref{LD-2})]
Set $m={m_J}=\{(1-\kappa)\wedge\frac{1}{{2}}\}\frac{J}{2}$ {and let $N=N_J$ be a positive integer depending on $J$ so that $\tau_J^wN\to \mathfrak{a}$ for some $w\in(0,\frac{1}{2})$ and $\mathfrak{a}\in(0,\infty)$}. {Then we define the set $I_{m,N}(i)$ by \eqref{notation-I}.} 
We decompose $E\left[\widehat{\rho}_{{J-j+1}}(\theta)|X\right]$ as
\begin{align*}
&E\left[\widehat{\rho}_{{J-j+1}}(\theta)|X\right]\\
&=\frac{\tau_J^{-1}}{n-l-L_j+1}\left\{\sum_{k=L_j-1}^{L_j+N-1}Z^1_kZ^2_{k+l}
+\sum_{k=2^{J-m}M_J}^{n-1-l}Z^1_kZ^2_{k+l}
+\sum_{i=0}^{M_J-1}\sum_{k\in {I_{m,N}(i)}}\left(Z^1_k-\sigma^1_{i\tau_m}\zeta^1_k\right)Z^2_{k+l}\right.\\
&\qquad+\sum_{i=0}^{M_J-1}\sum_{k\in {I_{m,N}(i)}}\sigma^1_{i\tau_m}\zeta^1_k\left(Z^2_{k+l}-\sigma^2_{i\tau_m+l\tau_J}\zeta^2_{k+l}\right)
+\sum_{i=0}^{M_J-1}\sigma^1_{i\tau_m}\sigma^2_{i\tau_m+l\tau_J}\sum_{k\in {I_{m,N}(i)}}\left(\zeta^1_k\zeta^2_{k+l}-E\left[\zeta^1_k\zeta^2_{k+l}\right]\right)\\
&\left.\qquad+\sum_{i=0}^{M_J-1}\sigma^1_{i\tau_m}\sigma^2_{i\tau_m+l\tau_J}\sum_{k\in {I_{m,N}(i)}}E\left[\zeta^1_k\zeta^2_{k+l}\right]\right\}\\
&=:\mathbb{I}_J(l)+\mathbb{II}_J(l)+\mathbb{III}_J(l)+\mathbb{IV}_J(l)+\mathbb{V}_J(l)+\mathbb{VI}_J(l),
\end{align*}
where $M_J=\lfloor 2^{m-J}(n-l-L_j-N)\rfloor$. 

First we prove $\max_{l\in\mathcal{L}_J^+}\left|\mathbb{I}_J(l)\right|\to^p0$. Since we have $\left\|Z^\nu_k\right\|_4\lesssim \sqrt{\tau_J}$ by the Minkowski and Burkholder-Davis-Gundy inequalities as well as \eqref{normalize}, we obtain
$
\left\|\mathbb{I}_J(l)\right\|_2
\lesssim\sum_{k=L_j-1}^{L_j+N-1}\left\|Z^1_k\right\|_4\left\|Z^2_{k+l}\right\|_4
\lesssim N\tau_J
$
by the triangle and Schwarz inequalities. Therefore, the Markov inequality yields
$
P\left(\max_{l\in\mathcal{L}_J^+}\left|\mathbb{I}_J(l)\right|>\varepsilon\right)
\leq\varepsilon^{-1}\sum_{l\in\mathcal{L}_J^+}\left\|\mathbb{I}_J(l)\right\|_2^2
\lesssim N^2\tau_J
$
for any $\varepsilon>0$, hence $\max_{l\in\mathcal{L}_J^+}\left|\mathbb{I}_J(l)\right|\to^p0$. 

Noting that $L^2\tau_J\to0$, we can prove $\max_{l\in\mathcal{L}_J^+}\left|\mathbb{II}_J(l)\right|\to^p0$ in an analogous manner to the above.

Next we prove $\max_{l\in\mathcal{L}_J^+}\left|\mathbb{III}_J(l)\right|\to^p0$. For any $k\in {I_{m,N}(i)}$ we have 
\begin{align*}
Z^1_k-\sigma^1_{i\tau_m}\zeta^1_k
&=(1-\pi_1)\left\{\sum_{\alpha=0}^{k\wedge N}\pi_1^\alpha\sum_{p=0}^{(L_j-1)\wedge(k-\alpha)}h_{j,p}\int_{(k-p-\alpha)\tau_J}^{(k-p-\alpha+1)\tau_J}(\sigma^1_s-\sigma^1_{i\tau_m})dB^1_s\right.\\
&\qquad\left.+\sum_{\alpha=N+1}^{k}\sum_{p=0}^{(L_j-1)\wedge(k-\alpha)}h_{j,p}\pi_1^\alpha\left(\Delta_{k-p-\alpha}X^1-\sigma^1_{i\tau_m}\Delta_{k-p-\alpha}B^1\right)\right\},
\end{align*}
hence it holds that
\begin{align*}
\left\|Z^1_k-\sigma^1_{i\tau_m}\zeta^1_k\right\|_r
&\lesssim\sqrt{\tau_J}\left\|\sup_{s\in[i\tau_m,(i+1)\tau_m+(L_j+N+1)\tau_J)}\left|\sigma^1_s-\sigma^1_{i\tau_m}\right|\right\|_r
+\sqrt{\tau_J}\pi_1^N\\
&\lesssim\sqrt{\tau_J}\left(\left(\tau_m+(L_j+N+1)\tau_J\right)^\gamma+\pi_1^N\right)
\end{align*}
for any $r\geq1$ by the Minkowski and Burkholder-Davis-Gundy inequalities as well as \eqref{holder}. Hence we obtain
\begin{align*}
P\left(\max_{l\in\mathcal{L}_J^+}\left|\mathbb{III}_J(l)\right|>\varepsilon\right)
&\lesssim\sum_{l\in\mathcal{L}_J^+}\left\{\sum_{i=0}^{M_J-1}\sum_{k\in {I_{m,N}(i)}}\left\|Z^1_k-\sigma^1_{i\tau_m}\zeta^1_k\right\|_{2r}\left\|Z^2_{k+l}\right\|_{2r}\right\}^r\\
&\lesssim\tau_J^{-1}\left\{M_J\cdot2^{J-m}\cdot\tau_J\left(\left(\tau_m+(L_j+N+1)\tau_J\right)^\gamma+\pi_1^N\right)\right\}^r\\
&=O\left(\tau_J^{-1}\left\{\left(\tau_m+(L_j+N+1)\tau_J\right)^\gamma+\pi_1^N\right\}^r\right)
\end{align*}
for any $\varepsilon>0$. We can take large enough $r\geq1$ such that $\tau_J^{-1}\left\{\left(\tau_m+(L_j+N+1)\tau_J\right)^\gamma+\pi_1^N\right\}^r\to0$, hence we obtain $\max_{l\in\mathcal{L}_J^+}\left|\mathbb{III}_J(l)\right|\to^p0$.

We can prove $\max_{l\in\mathcal{L}_J^+}\left|\mathbb{IV}_J(l)\right|\to^p0$ in an analogous manner.

Now we prove $\max_{l\in\mathcal{L}_J^+}\left|\mathbb{V}_J(l)\right|\to^p0$. 
{
By Lemma \ref{gqf} and the boundedness of $\sigma^1$ and $\sigma^2$, we have
\[
\|\mathbb{V}_J(l)\|_{\psi_1}\lesssim M_J\sqrt{\tau_m\tau_J}\left(L_j^{3/4}
+\sqrt{L_j(\pi_1^N+\pi_2^N)}
+\sqrt{L_j}|\log\tau_J|\right).
\]
Therefore, Lemma 2.2.2 of \cite{VW1996} yields
\[
\left\|\max_{l\in\mathcal{L}_J^+}\left|\mathbb{V}_J(l)\right|\right\|_{\psi_1}\lesssim|\log\tau_N|M_J\sqrt{\tau_m\tau_J}\left(L_j^{3/4}
+\sqrt{L_j(\pi_1^N+\pi_2^N)}
+\sqrt{L_j}|\log\tau_J|\right).
\]
Since $M_J\sqrt{\tau_m\tau_J}=O(\tau_J^{3/8})$ as $J\to\infty$ because $m\leq J/4$, we obtain
$\left\|\max_{l\in\mathcal{L}_J^+}\left|\mathbb{V}_J(l)\right|\right\|_{\psi_1}\to0$ by assumptions. This especially implies that $\max_{l\in\mathcal{L}_J^+}\left|\mathbb{V}_J(l)\right|\to^p0$. 
}

Finally, by Lemmas \ref{covariance}(b) and \ref{lemma:zero} we have $\max_{l\in\mathcal{L}_J^+}|\mathbb{VI}_J(l)|=o_p(M_J\cdot\tau_J^{-1}\tau_m\cdot\tau_J)$. Since $M_J=O(\tau_m^{-1})$, we obtain $\max_{l\in\mathcal{L}_J^+}|\mathbb{VI}_J(l)|\to^p0$. This completes the proof.  
\end{proof}

{\subsubsection{Completion of the proof of Theorem \ref{HRY}}}

{We need the following auxiliary result:} 
\begin{lemma}\label{nondegenerate}
$\int_{\Lambda_{-j}}D(\lambda)\Pi(\lambda)e^{\sqrt{-1}b\lambda}d\lambda\neq0$ for any $j\in\mathbb{N}$ and $b\in[-\frac{1}{2},\frac{1}{2}]$.
\end{lemma}

\begin{proof}
Since we have
$
\int_{\Lambda_{-j}}D(\lambda)\Pi(\lambda)e^{\sqrt{-1}b\lambda}d\lambda
=2\int_{\pi/2^j}^{\pi/2^{j-1}}D(\lambda)\Re\left[\Pi(\lambda)e^{\sqrt{-1}b\lambda}\right]d\lambda
$
and $D(\lambda)>0$ for any $\lambda\in\mathbb{R}$, it is enough to prove $\Re\left[\Pi(\lambda)e^{\sqrt{-1}b\lambda}\right]>0$ for any $\lambda\in(0,\pi)$. We have
\[
\Re\left[\Pi(\lambda)e^{\sqrt{-1}b\lambda}\right]
=\frac{(1-\pi_1)(1-\pi_2)}{|(1-\pi_1e^{\sqrt{-1}\lambda})(1-\pi_2e^{-\sqrt{-1}\lambda})|^2}\mathfrak{C},
\]
where $\mathfrak{C}=(1+\pi_1\pi_2)\cos b\lambda-\pi_1\cos(b-1)\lambda-\pi_2\cos(b+1)\lambda$. Hence it suffices to prove $\mathfrak{C}>0$. Since we have $\mathfrak{C}=(1+\pi_1\pi_2)\cos(-b\lambda)-\pi_1\cos((-b+1)\lambda)-\pi_2\cos((-b-1)\lambda)$, by symmetry we may assume $b\geq0$.

First we note that $\cos b\lambda>0$ because $b\lambda\in[0,\frac{\pi}{2})$. Next, we can rewrite $\mathfrak{C}$ as
\[
\mathfrak{C}=(1-\pi_1)(1-\pi_2)\cos b\lambda+\pi_1(\cos(-b)\lambda-\cos(b-1)\lambda)+\pi_2(\cos b\lambda-\cos(b+1)\lambda).
\]
Since $-\pi\leq(b-1)\lambda\leq(-b)\lambda\leq0$ due to $0\leq b\leq\frac{1}{2}$, we have $\cos(-b)\lambda-\cos(b-1)\lambda\geq0$ because $\cos$ is increasing on $[-\pi,0]$. Also, if $\lambda\geq\frac{\pi}{2}$, we have $\frac{\pi}{2}\leq(b+1)\lambda\leq\frac{3}{2}\pi$, hence $\cos(b+1)\lambda\leq0$. So $\cos b\lambda-\cos(b+1)\lambda\geq0$. Otherwise, we have $0\leq b\lambda\leq(b+1)\lambda\leq\frac{3}{4}\pi$, hence we have $\cos b\lambda-\cos(b+1)\lambda\geq0$ because $\cos$ is decreasing on $[0,\pi]$. Consequently, we have
$
\mathfrak{C}\geq(1-\pi_1)(1-\pi_2)\cos b\lambda>0.
$
This completes the proof.
\end{proof}

\begin{proof}[\upshape{\textbf{Proof of Theorem \ref{HRY}}}]
Suppose that there is a number $\varepsilon>0$ such that $P(v_J^{-1}|\widehat{\theta}_j-\theta_j|>\varepsilon)$ does not converge to 0 as $J\to\infty$. Then there is a sequence $(J_m)_{m\geq1}$ of positive integers such that $J_m\uparrow\infty$ as $m\to\infty$ and $P(v_{J_m}^{-1}|\widehat{\theta}_j-\theta_j|>\varepsilon)\to a$ as $m\to\infty$ for some $a>0$. Moreover, for every $m$ we can take an integer $l_m\in\mathcal{L}_{J_m}$ such that $|l_m\tau_{J_m}-\theta_j|\leq\tau_{J_m}/2$. In particular, the sequence $(\tau_{J_m}^{-1}(l_m\tau_{J_m}-\theta_j))_{m\geq1}$ has a converging subsequence. Without loss of generality we may assume that $\tau_{J_m}^{-1}(l_m\tau_{J_m}-\theta_j)\to b$ as $m\to\infty$ for some $b\in[-\frac{1}{2},\frac{1}{2}]$. Now since
$
\left|\widehat{\rho}_{{J-j+1}}(\widehat{\theta}_j)\right|>\max_{\theta\in\mathcal{G}_J:v_J^{-1}|\theta-\theta_j|> \varepsilon}\left|\widehat{\rho}_{{J-j+1}}(\theta)\right|
$
implies that $v_J^{-1}|\widehat{\theta}_j-\theta_j|\leq \varepsilon$, we have
\begin{align*}
P\left(v_{J_m}^{-1}\left|\widehat{\theta}_j-\theta_j\right|>\varepsilon\right)
&\leq P\left(\left|\widehat{\rho}_{{J-j+1}}(\widehat{\theta}_j)\right|\leq\max_{\theta\in\mathcal{G}_{J_m}:v_{J_m}^{-1}|\theta-\theta_j|> \varepsilon}\left|\widehat{\rho}_{{J-j+1}}(\theta)\right|\right)\\
&\leq P\left(\left|\widehat{\rho}_{{J-j+1}}(l_m\tau_{J_m})\right|\leq\max_{\theta\in\mathcal{G}_{J_m}:v_{J_m}^{-1}|\theta-\theta_j|> \varepsilon}\left|\widehat{\rho}_{{J-j+1}}(\theta)\right|\right)\\
&\leq P\left(\left|\widehat{\rho}_{{J-j+1}}(l_m\tau_{J_m})\right|\leq\frac{|\mathfrak{r}|}{2}\right)
+ P\left(\frac{|\mathfrak{r}|}{2}<\max_{\theta\in\mathcal{G}_{J_m}:v_{J_m}^{-1}|\theta-\theta_j|> \varepsilon}\left|\widehat{\rho}_{{J-j+1}}(\theta)\right|\right),
\end{align*}
where
$
\mathfrak{r}={2^j}\Sigma_{{T}}(\theta_j)R_{j}\int_{\Lambda_{-j}}D(\lambda)\Pi(\lambda)e^{\sqrt{-1}b\lambda}d\lambda.
$
{Now,} $\widehat{\rho}_{{J-j+1}}(l_m\tau_{J_m})\to^p\mathfrak{r}$ as $m\to\infty$ by {Theorem \ref{theorem:main}(b)}. Moreover, $\mathfrak{r}\neq0$ because of Lemma \ref{nondegenerate} and assumption. Therefore, by {Proposition} \ref{lemma:ld} we obtain 
\[
\limsup_{m\to\infty}P\left(v_{J_m}^{-1}\left|\widehat{\theta}_j-\theta_j\right|>\varepsilon\right)=0.
\]
This contradicts $\lim_{m\to\infty}P(v_{J_m}^{-1}|\widehat{\theta}_j-\theta_j|>\varepsilon)=a>0$.
\end{proof}

\section*{Acknowledgments}

We thank two anonymous referees for their careful reading of a previous version of the paper and valuable comments to it. 
Takaki Hayashi's research was supported by JSPS KAKENHI Grant Numbers JP16K03601, JP17H01100. 
Yuta Koike's research was supported by JST, CREST and JSPS KAKENHI Grant Number JP16K17105. 

{\small
\renewcommand*{\baselinestretch}{1}\selectfont
\addcontentsline{toc}{section}{References}

}

\end{document}